\documentclass[12pt]{article}

\usepackage[margin=1in]{geometry}

\usepackage{setspace}

\usepackage{graphicx}

\usepackage{enumitem}

\usepackage{microtype}
\usepackage{subfigure}
\usepackage{booktabs}
\usepackage{comment}
\usepackage{wrapfig}
\usepackage{arydshln}
\usepackage{enumitem}
\usepackage{multirow}
\usepackage{url}
\usepackage{natbib}
\usepackage{authblk}

\RequirePackage[colorlinks,citecolor=blue,linkcolor=blue,urlcolor=blue,pagebackref]{hyperref}

\usepackage{amsmath}
\usepackage{amssymb}
\usepackage{mathtools}
\usepackage{amsfonts}
\usepackage{bm}
\usepackage{bbm}

\usepackage{algorithm}
\usepackage{algorithmic}

\def\1{\bm{1}}

\DeclareMathAlphabet{\mathsfit}{\encodingdefault}{\sfdefault}{m}{sl}
\SetMathAlphabet{\mathsfit}{bold}{\encodingdefault}{\sfdefault}{bx}{n}

\def\calA{{\mathcal{A}}}

\def\calK{{\mathcal{K}}}

\def\calX{{\mathcal{X}}}
\def\calY{{\mathcal{Y}}}

\def\bbE{{\mathbb{E}}}

\def\bbR{{\mathbb{R}}}

\DeclareMathOperator*{\argmax}{arg\,max}

\newcommand{\p}[1]{\left(#1\right)}
\newcommand{\sqb}[1]{\left[#1\right]}
\newcommand{\cb}[1]{\left\{#1\right\}}

\usepackage{amsthm}

\theoremstyle{plain}

\newtheorem{theorem}{Theorem}[section]

\newtheorem{proposition}[theorem]{Proposition}

\newtheorem*{remark}{Remark}

\usepackage[textsize=tiny]{todonotes}
\usepackage{multirow}
\usepackage{wrapfig}
\usepackage{subfigure}

\usepackage{xcolor}
\newcount\Comments  
\newcommand{\kibitz}[2]{\ifnum\Comments=1\textcolor{#1}{#2}\fi}

\usepackage{listings}

\usepackage{tabularx}

\usepackage[capitalize,noabbrev]{cleveref}

\allowdisplaybreaks

\title{Vector Search As Nearest Neighbor Matching: RAG-based Policy Learning in Causal Inference}

\author{Masahiro Kato\thanks{Email: \texttt{mkato-csecon@g.ecc.u-tokyo.ac.jp}}$\,$}
\author{Taka Kato\thanks{Email: \texttt{taka@np-hard.co.jp}}$\,$}

\affil{${}^*$The University of Tokyo.}
\affil{${}^\dagger$NP-hard Inc.}

\date{\today}

\begin{document}

\maketitle

\begin{abstract}
We propose one-step and two-step methods for policy learning with retrieval-augmented generation (RAG). We formulate RAG-based action selection under the potential outcome framework. In the two-step method, vector search retrieves action-specific neighboring evidence in an embedding space, the generator estimates conditional expected outcomes or their contrasts, and a plug-in rule selects an action. This formulation connects action-specific vector search with nearest-neighbor matching in causal inference. We decompose the regret of the two-step method into candidate-generation regret and within-candidate choice regret, and we bound the latter using prediction-error guarantees for nearest-neighbor estimators and transformers. We evaluate the one-step method directly as a policy because its intermediate computation is unobserved.
\end{abstract}


\section{Introduction}
Retrieval-augmented generation (RAG) has become a common architecture for action selection in systems built on language models.
In a RAG algorithm, a query is embedded in a vector space, and documents or other information are retrieved according to distances in that space. Conditional on the retrieved evidence, a generative model produces an answer, which may include a recommendation or a tool choice.

Despite its widespread use, the decision-theoretic guarantees of RAG have not been fully studied.
If a query is treated as a covariate and the action recommended in the generated answer is treated as the decision, RAG-based decision-making can be viewed as a policy learning problem. Action selection and policy learning have long been studied in causal inference and reinforcement learning.

In causal inference, counterfactual decision-making has been studied as a policy learning problem \citep{Athey2021policylearning}. The goal is to recommend the action with the highest conditional expected outcome, so estimation of the conditional expected outcome is central. Two main approaches are the plug-in approach and the counterfactual risk minimization (CRM) approach \citep{Swaminathan2015counterfactualrisk}; the latter is also called empirical welfare maximization \citep{Kitagawa2018whoshould}. The plug-in approach first estimates the conditional expected outcome for each action and then chooses the action with the largest estimate. By contrast, the CRM approach estimates a value functional and learns a policy by maximizing the estimated value.

\subsection{Contents and Contributions}
We formulate RAG-based action selection under the potential outcome framework, also known as the Neyman--Rubin causal model. The explicit two-step procedure is a plug-in policy rule because it estimates action-specific conditional expected outcomes before choosing an action. We evaluate the one-step output as a policy without imposing the same interpretation on its unobserved internal computation.

We first consider the case in which $K$ actions are given. Let $X$ denote covariates that represent the query or decision context supplied to the RAG system, let $A\in\calK\coloneqq\cb{0,1,\dots,K-1}$ denote the action, and let $Y(a)$ denote the potential outcome under action $a$. Our goal is to choose the action with the highest conditional expected outcome, and we write the context-specific optimal action as
\[a^*(x)\in\argmax_{a\in\calK}\bbE\sqb{Y(a)\mid X=x}.\]
Because vector search retrieves documents or examples that are similar or relevant to the query, we interpret the RAG system as producing an estimate of $a^*(x)$ from the context $x$. When the retrieval database contains covariates, actions, and outcomes, vector search supplies nearest neighbors used to estimate the conditional expected outcomes that determine $a^*(x)$. Similarity is therefore used to select matched evidence for each action, rather than to choose the action that appears most often among neighboring cases. This operation corresponds to nearest-neighbor matching in causal inference.

We next extend this formulation to settings in which the available
action set is not specified in advance. Although the preceding formulation assumes that $K$ actions are given, the set of available actions is often not known in advance in applications of language models. For example, in workplace applications, a user may submit a query and ask what to do next without providing a candidate action set. To reflect this setting, we propose two forms of RAG-based policy learning. The one-step method receives a query and directly returns a recommended action, while the actions considered internally remain unobserved. The two-step method proceeds as follows:
\begin{itemize}
    \item We provide a query and ask the RAG system to generate a finite set of candidate actions.
    \item For each candidate action, the RAG system retrieves matched evidence and evaluates or ranks the action according to its estimated conditional expected outcome.
    \item We select the action with the largest estimated conditional expected outcome or the highest rank.
\end{itemize}
The two-step method separates action generation from action choice. This separation allows us to decompose its regret into the loss from the generated candidate set and the loss from expected-outcome estimation or ranking within that set.

This study makes three contributions:
\begin{itemize}
    \item We formulate RAG-based action selection as policy learning under the potential outcome framework and give an explicit plug-in interpretation for the two-step procedure.
    \item We interpret action-specific vector search over observed cases as nearest-neighbor matching and distinguish evidence retrieval from action choice.
    \item We decompose regret into candidate-set regret and within-candidate regret, and we bound the latter using prediction-error guarantees.
\end{itemize}

\subsection{Related Work}
This study is related to three streams of literature: policy learning, RAG, and nonparametric analysis of neural networks and transformers.

\paragraph{Policy learning.}
We study policy learning in causal inference \citep{Swaminathan2015counterfactualrisk,Kitagawa2018whoshould,Athey2021policylearning}. As in supervised learning \citep{Audibert2011fast}, policy learning can be approached through plug-in estimation or counterfactual risk minimization. The latter is also referred to as empirical welfare maximization (EWM) \citep{Kitagawa2018whoshould}, and a theoretical comparison of the two approaches is given by \citet{Kitagawa2018whoshould}. Although that comparison gives favorable results for EWM relative to the plug-in approach under its conditions, their relative performance depends on the underlying data-generating process, as discussed by \citet{Audibert2011fast}. Because the structure of RAG is more directly compatible with plug-in estimation, we formulate RAG-based policy learning using the plug-in approach.

In particular, we interpret action-specific retrieval over observed cases as a nearest-neighbor matching procedure. Nearest-neighbor match counts have density-ratio limits that are proportional to inverse propensity scores, and nearest-neighbor matching can also be related to Riesz regression \citep{Lin2023estimationbased,Kato2025nearestneighbor}. The analyses of \citet{Kitagawa2018whoshould} and \citet{Athey2021policylearning} use inverse probability weighting (IPW) or augmented IPW estimators. The representation results in \citet{Lin2023estimationbased} and \citet{Kato2025nearestneighbor} therefore connect RAG-based policy learning to those analyses.

\paragraph{RAG.}
Our study is also related to work on RAG. \citet{Guu2020realm} jointly trains a latent document retriever and a language model. \citet{Lewis2020retrievalaugmentedgeneration} combines a dense passage index with a sequence-to-sequence generator and marginalizes over retrieved passages. Dense passage retrieval maps queries and passages into a shared embedding space and ranks passages by inner-product similarity \citep{Karpukhin2020densepassage}. This retrieval step is computationally related to nearest-neighbor methods. In particular, the $k$-nearest-neighbor language model constructs a nonparametric next-token distribution from nearby contextual representations and interpolates it with the distribution produced by a parametric language model \citep{Khandelwal2020generalizationthrough}. Later studies develop different ways to incorporate retrieved information into generation. Fusion-in-Decoder encodes retrieved passages separately and combines their representations in the decoder \citep{Izacard2021leveragingpassage}. REPLUG augments a frozen black-box language model and trains the retriever using feedback from the language model \citep{Shi2024replug}. Self-RAG learns when to retrieve and uses self-reflection signals to assess whether the retrieved passages support the generated response \citep{Asai2024selfrag}. These methods select external evidence according to proximity in a representation space, but they do not interpret the retrieved items as matched observations under the potential outcome framework.

Retrieval has also been incorporated into sequential decision-making. \citet{Goyal2022retrievalaugmentedreinforcement} augments reinforcement learning agents with direct access to datasets of past experience. \citet{Humphreys2022largescaleretrieval} uses approximate nearest-neighbor search over a large collection of expert states to support offline reinforcement learning. REGENT conditions a generalist policy on demonstrations retrieved from related tasks \citep{Sridhar2025regent}, whereas STRAP retrieves relevant sub-trajectories before learning a policy for a target task \citep{Memmel2024strap}. These studies place retrieval within reinforcement learning or imitation learning rather than plug-in policy learning under the potential outcome framework. In a different direction, CausalRAG incorporates causal graphs into retrieval and generation for knowledge-intensive tasks \citep{Wang2025causalrag}. Its use of causal information concerns relations among concepts in the retrieved documents rather than counterfactual outcomes under alternative actions.

Unlike these approaches, we formulate RAG-based action selection under the potential outcome framework. For the two-step method, vector search supplies local evidence for outcome regression, the generator estimates action-specific conditional expected outcomes or their contrasts, and the algorithm selects the action with the largest estimate. The one-step method is evaluated at the level of its returned policy because its internal criterion is not observed. The two-step formulation also separates regret due to candidate generation from regret due to expected-outcome estimation or ranking within the generated set.

\paragraph{Nonparametric analysis of neural networks and transformers.}
The statistical theory of neural networks provides a basis for viewing transformers as nonparametric estimators. Nonparametric regression with deep ReLU networks has been studied under compositional assumptions and over Besov-type function classes \citep{SchmidtHieber2020nonparametric,Suzuki2019adaptivityof,Suzuki2021deeplearning}. For transformers, \citet{Yun2020aretransformers} establishes universal approximation results for sequence-to-sequence functions, while \citet{Takakura2023approximationand} and \citet{Havrilla2024understandingscaling} examine approximation and estimation for high-dimensional sequence inputs and data with low-dimensional structure.

A related literature studies in-context learning as a statistical learning procedure. Existing work examines the function classes that transformers can learn from in-context examples and relates their predictions to regression algorithms, gradient descent, and algorithm selection \citep{Garg2022whatcan,Akyurek2023whatlearning,Oswald2023transformerslearn,Bai2023transformersas}. More recent studies analyze in-context learning for nonparametric regression and adaptation to low-dimensional target functions \citep{Kim2024transformersare,Oko2024pretrainedtransformer,Ching2026efficientand}. We use prediction-error guarantees from this literature as inputs to the regret analysis for transformer-based expected-outcome estimation.

\section{Setup}
Let $X\in\calX\subseteq\mathbb R^d$ denote covariates, which represent a query or decision context. The covariates can combine the query with pre-action individual attributes, such as age, gender, and occupation, when those attributes are available and relevant to the decision problem. Let $A\in\calA$ denote an action, where the action space $\calA$ can be finite or uncountably infinite. Let $Y\in\calY\subseteq\bbR$ denote a scalar outcome. We denote the conditional expected outcome of $Y$ given $A=a\in\calA$ and $X=x\in\calX$ by $f_0(a,x)\coloneqq\bbE\sqb{Y\mid A=a,X=x}$.

\subsection{Policy and Policy Value}
We focus on deterministic policies. A policy is a measurable function
\[
\pi\colon\calX\to\calA.
\]
Let $\Pi$ be a class of such policies. For $\pi\in\Pi$, define its value by
\[
V(\pi)
\coloneqq
\bbE\sqb{f_0(\pi(X),X)}.
\]

We denote a learned or otherwise data-dependent policy by
$\widehat a\colon\calX\to\calA$, so that $\widehat a(x)$ is the action
selected at covariate value $x$. Its value is
\[
V(\widehat a)
\coloneqq
\bbE\sqb{f_0(\widehat a(X),X)}.
\]
When $\widehat a$ depends on training data, the retrieval database, or
auxiliary algorithmic randomness, the expectation also averages over
this randomness.

For a chosen action $\widehat a(\cdot)$, we define its value by
\[
V(\widehat a)\coloneqq\bbE\sqb{f_0(\widehat a(X),X)}.
\]

\subsection{Optimal Policy}
We assume that the maximum of $f_0(a,x)$ over $a\in\calA$ is attained
and that a measurable maximizer can be selected. For each $x\in\calX$,
let
\[
a^*(x)
\in
\argmax_{a\in\calA} f_0(a,x).
\]
The measurable selector $a^*\colon\calX\to\calA$ is the optimal
deterministic policy. Its value is
\[
V(a^*)
=
\bbE\sqb{\max_{a\in\calA} f_0(a,X)}.
\]

\subsection{Regret}
The regret of $\widehat a$ relative to the optimal action rule $a^*$ is
\[
R(\widehat a)
\coloneqq
V(a^*)-V(\widehat a)
=
\bbE\sqb{
f_0(a^*(X),X)-f_0(\widehat a(X),X)
}.
\]

\section{RAG-PL}
\label{sec:rag-pl}
In this section, we describe our proposed method, called RAG-based policy learning (RAG-PL). We consider two forms of RAG-PL. The one-step method returns an action directly from the covariates. The two-step method first generates candidate actions and then selects one action from those candidates. The two-step method has two stages. The first stage constructs the candidate set. The second stage evaluates the candidates using retrieved evidence and selects one of them.

\subsection{RAG}
We consider three input-output forms of a RAG system. The first returns a set of candidate actions given the covariates. he second returns an estimated conditional expected outcome for each action or ranks the candidate actions by that quantity. The third returns a recommended action directly from the covariates.

\paragraph{Action-set RAG.}
Let $\mathfrak F(\calA)\coloneqq\cb{C\subseteq\calA:1\leq |C|<\infty}$ denote the collection of all non-empty finite subsets of $\calA$. We represent an action-set RAG as a set-valued function $g\colon\calX\to\mathfrak F(\calA)$. For each $x\in\calX$, the set $g(x)\subseteq\calA$ contains the candidate actions returned by the RAG system.

\paragraph{Vector search as nearest-neighbor matching.}
Let $\varphi\colon\calX\to\mathbb R^{d_\varphi}$ be the embedding map used for vector search, and write $H=\varphi(X)$. Let $\mathcal D=\cb{D_j}_{j=1}^{N_{\mathrm{db}}}$ be a retrieval database. When $\calA$ is discrete and the database contains observations of
covariates, actions, and outcomes, we write
$D_j=(X_j,A_j,Y_j)$ and, for an action $a$, define the action-specific
index set
$\mathcal I_a\coloneqq\cb{j:A_j=a}$. Given $x$ and $a$, let $\mathcal N_{k,a}(x)$ be the indices of the $k$ observations indexed by $\mathcal I_a$ whose embeddings are closest to $\varphi(x)$, and let $\mathcal R_k(x,a)\coloneqq\cb{D_j:j\in\mathcal N_{k,a}(x)}$ denote the retrieved evidence. We assume that $|\mathcal I_a|\geq k$ for every action evaluated through action-specific retrieval. This operation matches the current context with similar cases within each action, as in nearest-neighbor matching. For a general document corpus, $\mathcal R_k(x,a)$ denotes evidence retrieved by a query that includes the candidate action. In either case, similarity determines which evidence is used to evaluate an action. It is not itself the criterion for choosing the final action, and we do not select an action by the frequency with which it appears among neighboring cases. For continuous action spaces, exact action-specific retrieval based on
$\mathcal I_a=\cb{j:A_j=a}$ is generally unavailable. In that case,
retrieval must use an action-aware distance or kernel on
$\calA\times\calX$, or another procedure that borrows information across
nearby actions. The nearest-neighbor analysis in
Section~\ref{sec:nn-regret} is restricted to binary actions.

\paragraph{Expected-outcome RAG or ranking RAG.}
We represent an expected-outcome RAG by a function $\widehat f\colon\calA\times\calX\to\bbR$. Given an action $a\in\calA$ and covariates $x\in\calX$, it returns an estimate $\widehat f(a,x)$ of the conditional expected outcome $f_0(a,x)$. When the generator and retrieved evidence are written explicitly, we use
\[
\widehat f(a,x)
=
G_\theta\bigl(x,a,\mathcal R_k(x,a)\bigr).
\]
A ranking RAG instead receives $x$ and a finite candidate set $C=\cb{a_1,\dots,a_m}\in\mathfrak F(\calA)$. It retrieves evidence for the candidate actions and returns an ordered tuple $r(x,C)=(a_{(1)},\dots,a_{(m)})$, which is a permutation of the actions in $C$ from the highest to the lowest estimated conditional expected outcome.

\begin{remark}[Fixed action set]
    When the action set is given and finite, we write $\calA=\calK\coloneqq\cb{0,1,\dots,K-1}$. These actions are also called treatments or arms. In the potential-outcome notation, each unit has potential outcomes $Y(0),Y(1),\dots,Y(K-1)$. We use the special case $K=2$ when discussing the relation between the selected action and the conditional average treatment effect.
\end{remark}

\paragraph{Policy RAG.}
A policy RAG receives covariates and directly returns a recommended action. We represent it by a function $\pi_{\mathrm{RAG}}\colon\calX\to\calA$.

\subsection{One-Step RAG-PL}
The one-step RAG-PL method applies a policy RAG directly. Given $x\in\calX$, the RAG system is instructed to recommend an action $\widehat a_{\mathrm{one}}(x)$ with the highest conditional expected outcome, and we write $\widehat a_{\mathrm{one}}(x)=\pi_{\mathrm{RAG}}(x)$. The policy implemented by the RAG system is usually not observed separately from its output. In particular, the one-step method does not expose an intermediate candidate set, expected outcomes for the candidate actions, or a separate ranking step. We therefore evaluate the returned policy without imposing a particular interpretation on its internal computation.

\subsection{Two-Step RAG-PL}

\paragraph{Action-set generation.}
Given $x\in\calX$, the action-set RAG returns a finite candidate set
\[
g(x)=\cb{a_1(x),\dots,a_{M(x)}(x)}\subseteq\calA,
\qquad 1\leq M(x)<\infty.
\]
Thus, even when $\calA$ is uncountably infinite, the subsequent evaluation or ranking is performed only over the finite set $g(x)$. The candidate set can contain actions found in matched cases and actions generated by the language model. Its purpose is to include actions with high conditional expected outcomes, rather than to reproduce the empirical distribution of actions in the retrieved evidence.

\paragraph{Ranking or expected-outcome estimation.}
For every $a\in g(x)$, the RAG system retrieves $\mathcal R_k(x,a)$. When an expected-outcome RAG is used, it evaluates $\widehat f(a,x)=G_\theta(x,a,\mathcal R_k(x,a))$. When a ranking RAG is used, it returns the ordered tuple $r(x,g(x))=(a_{(1)}(x),\dots,a_{(M(x))}(x))$, where $a_{(1)}(x)$ is the candidate with the highest estimated conditional expected outcome.

\paragraph{Action choice.}
Under expected-outcome estimation, we choose
\[
\widehat a_{\mathrm{two}}(x)\in\argmax_{a\in g(x)}\widehat f(a,x),
\]
where a fixed rule is used to break ties. Under ranking, we choose
\[
\widehat a_{\mathrm{two}}(x)\coloneqq a_{(1)}(x).
\]
Both methods define a policy from $\calX$ to $\calA$. The two-step method makes action generation, retrieval of similar cases, and action choice explicit, which allows their contributions to regret to be studied separately.

\section{Regret Analysis}
\label{sec:regret-analysis}
We analyze the regret of the one-step and two-step methods under the setup above. We write $R(\widehat a)$ for the regret of a selected action $\widehat a(\cdot)$. When the candidate set, the expected-outcome estimate, or the selected action is random, the expectation defining $R(\widehat a)$ and all expectations below also average over this randomness. We first separate the effect of the generated candidate set from the effect of selecting an action within that set. We then connect the second term to nonparametric prediction error. Section~\ref{sec:nn-regret} studies the case in which vector search is represented explicitly by nearest-neighbor matching.

\subsection{Regret Decomposition for Two-Step RAG-PL}
For each $x\in\calX$, let
\[
a_g^*(x)\in\argmax_{a\in g(x)}f_0(a,x),
\]
where the fixed tie-breaking rule is used when needed, and we assume that the resulting selector is measurable. The maximum is attained because $g(x)$ is finite. We define
\[
R_{\mathrm{gen}}(g)
\coloneqq
\bbE\sqb{f_0(a^*(X),X)-f_0(a_g^*(X),X)}
\]
and
\[
R_{\mathrm{choice}}(\widehat a_{\mathrm{two}};g)
\coloneqq
\bbE\sqb{f_0(a_g^*(X),X)-f_0(\widehat a_{\mathrm{two}}(X),X)}.
\]
The first term is the regret caused by restricting the choice to $g(X)$. The second is the regret caused by expected-outcome estimation or ranking within $g(X)$. We call $R_{\mathrm{gen}}(g)$ the candidate-set regret and $R_{\mathrm{choice}}(\widehat a_{\mathrm{two}};g)$ the
within-candidate regret.

Adding and subtracting $f_0(a_g^*(X),X)$ gives
\begin{equation}
\label{eq:two-step-decomposition}
R(\widehat a_{\mathrm{two}})
=
R_{\mathrm{gen}}(g)
+
R_{\mathrm{choice}}(\widehat a_{\mathrm{two}};g).
\end{equation}
This identity does not require independence between candidate generation and action choice. 

\subsection{Bounding the Candidate-Set Regret}
For $\varepsilon\geq0$, define the set of $\varepsilon$-optimal actions by
\[
\calA_\varepsilon^*(x)
\coloneqq
\cb{a\in\calA:f_0(a^*(x),x)-f_0(a,x)\leq\varepsilon}.
\]
Exact inclusion of $a^*(x)$ is not required. It is enough that the generated set contain an action whose conditional expected outcome is close to the optimum.

\begin{proposition}[Bounds for the generated candidate set]
\label{prop:candidate-bounds}
The following statements hold.
\begin{enumerate}
    \item If $g(x)\cap\calA_\varepsilon^*(x)\neq\varnothing$ holds for $P_X$-almost every $x$, then we have
    \[
    R_{\mathrm{gen}}(g)\leq\varepsilon.
    \]
    \item Suppose that $0\leq f_0(a^*(x),x)-f_0(a,x)\leq B_f$ for every $(a,x)\in\calA\times\calX$. If
    \[
    \Pr\p{g(X)\cap\calA_\varepsilon^*(X)=\varnothing}\leq\delta
    \]
    holds, then we have
    \[
    R_{\mathrm{gen}}(g)\leq\varepsilon+B_f\delta.
    \]
    \item Let $d_\calA$ be a metric on $\calA$, and define $d_\calA(a,C)\coloneqq\min_{b\in C}d_\calA(a,b)$ for every non-empty finite set $C$. Suppose that, for some $L_\calA>0$ and $\beta_\calA>0$,
    \[
    |f_0(a,x)-f_0(b,x)|
    \leq
    L_\calA d_\calA(a,b)^{\beta_\calA}
    \]
    holds for every $a,b\in\calA$ and $x\in\calX$. Then, we have
    \[
    R_{\mathrm{gen}}(g)
    \leq
    L_\calA\bbE\sqb{d_\calA(a^*(X),g(X))^{\beta_\calA}}.
    \]
\end{enumerate}
\end{proposition}

The third statement controls the value of the best action in $g(x)$ through its distance from $a^*(x)$. It does not concern the estimation error of $\widehat f$. The smoothness condition on $f_0$ converts distance in the action space into a difference in conditional expected outcomes.

\subsection{Expected-Outcome Estimation and Ranking}
Suppose that the two-step method chooses
\[
\widehat a_{\mathrm{two}}(x)
\in
\argmax_{a\in g(x)}\widehat f(a,x).
\]
For each $x\in\calX$, define
\[
\delta_f(x;g)
\coloneqq
\max_{a\in g(x)}|\widehat f(a,x)-f_0(a,x)|.
\]
Let $v_g^*(x)\coloneqq\max_{a\in g(x)}f_0(a,x)$ and define the smallest positive gap within the candidate set by
\[
\Delta_g(x)
\coloneqq
\min_{a\in g(x):f_0(a,x)<v_g^*(x)}
\cb{v_g^*(x)-f_0(a,x)},
\]
where $\Delta_g(x)=+\infty$ if every action in $g(x)$ attains $v_g^*(x)$.

\begin{theorem}[Regret from expected-outcome estimation]
\label{thm:expected-outcome-regret}
For every $x\in\calX$, it holds that
\[
f_0(a_g^*(x),x)-f_0(\widehat a_{\mathrm{two}}(x),x)
\leq
2\delta_f(x;g).
\]
Consequently,
\[
R(\widehat a_{\mathrm{two}})
\leq
R_{\mathrm{gen}}(g)
+
2\bbE\sqb{\delta_f(X;g)}.
\]
Suppose in addition that $\delta_f(X;g)\leq r_f$ almost surely and that, for some $C_M>0$, $\kappa\geq0$, and $t_0>0$,
\[
\Pr\p{0<\Delta_g(X)\leq t}
\leq
C_Mt^\kappa,
\qquad
0<t\leq t_0.
\]
If $2r_f\leq t_0$, then
\[
R_{\mathrm{choice}}(\widehat a_{\mathrm{two}};g)
\leq
C_M(2r_f)^{1+\kappa}.
\]
\end{theorem}

The theorem is stated for the realized set $g(X)$ and therefore does not require candidate generation to be independent of $\widehat f$. A high-probability error bound can be used in the same way. If $0\leq f_0(a_g^*(X),X)-f_0(a,X)\leq B_f$ holds almost surely for every $a\in g(X)$, the margin condition holds, and
\[
\Pr\p{\delta_f(X;g)>r_f}\leq\eta_f,
\]
then
\begin{equation}
\label{eq:high-prob-choice}
R_{\mathrm{choice}}(\widehat a_{\mathrm{two}};g)
\leq
C_M(2r_f)^{1+\kappa}+B_f\eta_f.
\end{equation}
For example, suppose that $|g(X)|\leq M$ and, conditional on $X$ and $g(X)$,
\[
\Pr\left(
|\widehat f(a,X)-f_0(a,X)|>b_n+u
\mathrel{\Big|}
X,g(X)
\right)
\leq
c_1\exp(-c_2a_nu^2)
\]
for every $a\in g(X)$. A union bound gives
\[
\Pr\p{\delta_f(X;g)>b_n+u}
\leq
c_1M\exp(-c_2a_nu^2).
\]
Thus, the number of candidates enters the simultaneous error through a logarithmic term when $u$ is chosen to make the right-hand side small.

For a ranking RAG, let $a_{(1)}(x)$ be the first action in $r(x,g(x))$. Then,
\[
R(a_{(1)})
=
R_{\mathrm{gen}}(g)
+
\bbE\sqb{f_0(a_g^*(X),X)-f_0(a_{(1)}(X),X)}.
\]
When the ranking orders the candidates according to $\widehat f(a,x)$, Theorem~\ref{thm:expected-outcome-regret} applies. If the RAG system returns only an ordering, its analysis requires a direct bound on the true conditional expected-outcome difference between the best candidate and the first-ranked candidate.

\subsection{Relation to Policy-Value Maximization}
Let
\[
\mathcal S_g
\coloneqq
\cb{s\colon\calX\to\calA:s\text{ is measurable and }s(x)\in g(x)}
\]
and define the value induced by $\widehat f$ as
\[
\widehat V_f(s)
\coloneqq
\bbE\sqb{\widehat f(s(X),X)}.
\]
If a measurable selector $\widehat s_f$ satisfies
\[
\widehat s_f(x)
\in
\argmax_{a\in g(x)}\widehat f(a,x),
\]
then, for every $s\in\mathcal S_g$, it holds that $\widehat f(\widehat s_f(x),x)\geq\widehat f(s(x),x)$. Hence,
\[
\widehat s_f
\in
\argmax_{s\in\mathcal S_g}\widehat V_f(s).
\]
The same pointwise argument applies to an empirical average over target contexts. Thus, when the policy class contains every measurable selector and its value is constructed from $\widehat f$, pointwise maximization and maximization over the policy class are the same optimization problem. They can differ when the policy class is restricted or when the value estimator uses observed outcomes through IPW or AIPW rather than being constructed from $\widehat f$ alone.

\subsection{Substitution of Nonparametric Rates}
The preceding results allow a prediction rate to be inserted into the regret bound without changing the RAG-PL algorithm. Suppose that $R_{\mathrm{gen}}(g)\leq r_{g,n}$ and that the expected-outcome error on the generated set is bounded by $r_{f,n}$. Under the finite-candidate margin condition in Theorem~\ref{thm:expected-outcome-regret}, we have
\[
R(\widehat a_{\mathrm{two}})
\leq
r_{g,n}+C_M(2r_{f,n})^{1+\kappa}.
\]
For example, if a nonparametric estimator has a uniform error bound $r_{f,n}\asymp n^{-\alpha/(2\alpha+d)}$, the second term is of order $n^{-\alpha(1+\kappa)/(2\alpha+d)}$. This substitution is valid only when the available prediction result controls the error required by the regret theorem. An integrated mean-squared error cannot be used as a uniform error without an additional argument.

The same distinction matters for transformers. The analysis of \citet{Kim2024transformersare} gives a mean-squared prediction bound for a transformer composed of a learned neural representation and a linear-attention layer. In the Besov setting, a representative bound has the form
\[
q_{\mathrm{ICL}}
\lesssim
N_{\mathrm{rep}}^{-2\alpha/d}
+
\frac{N_{\mathrm{rep}}\log N_{\mathrm{rep}}}{n_{\mathrm{ctx}}}
+
\frac{N_{\mathrm{rep}}^2\log N_{\mathrm{rep}}}{T_{\mathrm{pre}}},
\]
where $N_{\mathrm{rep}}$ is the representation dimension, $n_{\mathrm{ctx}}$ is the number of in-context examples, and $T_{\mathrm{pre}}$ is the number of pretraining tasks. If this bound holds for both binary expected-outcome estimates, Theorem~\ref{thm:binary-regret} below gives
\[
R(\widehat a)
\lesssim
(q_{\mathrm{ICL},0}+q_{\mathrm{ICL},1})^{(1+\kappa)/(2+\kappa)}.
\]
A sharper exponent $1+\kappa$ for the root prediction rate requires a uniform error bound or a pointwise exponential-deviation bound. It does not follow from an MSE result alone. The results of \citet{Oko2024pretrainedtransformer} and \citet{Ching2026efficientand} further describe adaptation to low-dimensional target functions and minimax MSE rates, and they enter the regret analysis according to the same distinction between error types.

The general regret decomposition itself has no sample size. For a specified retrieval database, its size and the number of matched observations have a direct statistical meaning. Transformer theory uses different quantities for the representation dimension, in-context examples, and pretraining tasks. For an already trained RAG system whose training data are not observed, the total number of pretraining tokens is not by itself a sample size for conditional expected-outcome estimation. In that case, the regret bound should be stated in terms of a prediction error that is either assumed or evaluated separately.

\subsection{Relation Between One-Step and Two-Step RAG-PL}
The regret of the one-step output is
\[
R(\widehat a_{\mathrm{one}})
=
\bbE\sqb{f_0(a^*(X),X)-f_0(\widehat a_{\mathrm{one}}(X),X)}.
\]
For comparison, its observed output can be represented by the singleton set $g_{\mathrm{one}}(x)=\cb{\widehat a_{\mathrm{one}}(x)}$. This representation does not assert that the RAG system uses a singleton set internally.

Suppose that $\widehat a_{\mathrm{one}}(x)\in g(x)$ for $P_X$-almost every $x$. Because $a_g^*(x)$ is at least as good as every action in $g(x)$, we have
\[
R_{\mathrm{gen}}(g)
\leq
R(\widehat a_{\mathrm{one}}).
\]
Combining this inequality with \eqref{eq:two-step-decomposition} gives
\[
R(\widehat a_{\mathrm{two}})
\leq
R(\widehat a_{\mathrm{one}})
+
R_{\mathrm{choice}}(\widehat a_{\mathrm{two}};g).
\]
If the two-step method selects the best action in $g(x)$ according to $f_0$, its regret is no larger than the regret of the one-step output. With an estimated expected outcome, Theorem~\ref{thm:expected-outcome-regret} instead gives
\[
R(\widehat a_{\mathrm{two}})
\leq
R(\widehat a_{\mathrm{one}})
+
2\bbE\sqb{\delta_f(X;g)}.
\]
Thus, a larger candidate set can contain a better action than the one-step output, but the gain can be lost when expected-outcome estimation or ranking is inaccurate.

\subsection{Fixed Action Sets and Binary Actions}
\label{subsec:binary-regret}
When the action set is fixed and finite, we can set $g(x)=\calA=\calK$ for every $x$. Then, $R_{\mathrm{gen}}(g)=0$, and Theorem~\ref{thm:expected-outcome-regret} gives
\[
R(\widehat a_{\mathrm{two}})
\leq
2\bbE\sqb{\max_{a\in\calK}|\widehat f(a,X)-f_0(a,X)|}.
\]

For the binary action set $\calK=\cb{0,1}$, define
\[
\tau_0(x)
\coloneqq
f_0(1,x)-f_0(0,x),
\qquad
\widehat\tau(x)
\coloneqq
\widehat f(1,x)-\widehat f(0,x).
\]
Using the tie-breaking rule that selects action $1$ at equality, we have $a^*(x)=\mathbbm{1}\sqb{\tau_0(x)\geq0}$ and $\widehat a(x)=\mathbbm{1}\sqb{\widehat\tau(x)\geq0}$. For any selected action $\widehat a(x)\in\cb{0,1}$,
\begin{equation}
\label{eq:binary-regret-identity}
R(\widehat a)
=
\bbE\sqb{|\tau_0(X)|\mathbbm{1}\sqb{\widehat a(X)\neq a^*(X)}}.
\end{equation}
When $f_0(a,x)=\bbE\sqb{Y(a)\mid X=x}$, $\tau_0(x)$ is the conditional average treatment effect.

We use the margin condition
\begin{equation}
\label{eq:binary-margin}
\Pr\p{0<|\tau_0(X)|\leq t}
\leq
C_Mt^\kappa,
\qquad
0<t\leq t_0.
\end{equation}
It controls the probability of contexts for which the two actions have nearly equal conditional expected outcomes \citep{Audibert2011fast}.

\begin{theorem}[Binary regret bounds]
\label{thm:binary-regret}
Suppose that \eqref{eq:binary-margin} holds.
\begin{enumerate}
    \item If $\lVert\widehat\tau-\tau_0\rVert_\infty\leq r_\tau$ holds almost surely and $r_\tau\leq t_0$, then
    \[
    R(\widehat a)
    \leq
    C_Mr_\tau^{1+\kappa}.
    \]
    In particular, if $\max_{a\in\cb{0,1}}|\widehat f(a,x)-f_0(a,x)|\leq r_f$ for $P_X$-almost every $x$, then $R(\widehat a)\leq C_M(2r_f)^{1+\kappa}$ whenever $2r_f\leq t_0$.
    \item Suppose that $|\tau_0(X)|\leq B_\tau$ holds almost surely and that there exist $b_n\geq0$, $a_n>0$, and $\delta_n\geq0$, together with constants $c_1,c_2>0$, such that
    \[
    \Pr\p{|\widehat\tau(X)-\tau_0(X)|\geq b_n+u\mid X=x}
    \leq
    c_1\exp(-c_2a_nu^2)+\delta_n
    \]
    holds for $P_X$-almost every $x\in\calX$ and every $u>0$. If $b_n\leq t_0/4$ and $a_n^{-1/2}\leq t_0/4$, then
    \[
    R(\widehat a)
    \leq
    C(b_n+a_n^{-1/2})^{1+\kappa}
    +
    B_\tau\delta_n.
    \]
    \item If 
    \[
    \bbE\sqb{(\widehat\tau(X)-\tau_0(X))^2}
    \leq
    q_n
    \]
    and $q_n^{1/(2+\kappa)}\leq t_0$ hold, then
    \[
    R(\widehat a)
    \leq
    Cq_n^{(1+\kappa)/(2+\kappa)}.
    \]
\end{enumerate}
The constants in the second and third statements depend only on the constants displayed in their assumptions and the margin condition.
\end{theorem}

\section{Nearest-Neighbor Matching and Expected-Outcome Regret}
\label{sec:nn-regret}
Vector search connects RAG-PL with causal inference by selecting observations whose covariates are close to the current context. This section gives an explicit nearest-neighbor regression model in which the expected-outcome estimate in Section~\ref{sec:rag-pl} is formed from matched observations and its error can be substituted into the regret bounds above.

\subsection{Nearest-Neighbor Matching}
For each $a\in\cb{0,1}$, let
\[
\mathcal D_a
=
\cb{(H_{a,j},Y_{a,j})}_{j=1}^{N_a}
\]
be a database for action $a$, where $H_{a,1},\dots,H_{a,N_a}$ are independently and identically distributed according to $Q_a$ on $\mathbb R^{d_\varphi}$ and
\[
Y_{a,j}=m_a(H_{a,j})+\xi_{a,j}.
\]
We assume that the query $X$ is independent of the action-specific databases. For a query $x$, write $h=\varphi(x)$ and let $\mathcal N_{k,a}(h)$ be the indices of the $k$ nearest embeddings in $\mathcal D_a$. Define the matched expected-outcome estimate by
\[
\widetilde f_k(a,x)
\coloneqq
\frac{1}{k}\sum_{j\in\mathcal N_{k,a}(h)}Y_{a,j}.
\]
The following result gives a standard bias and variance calculation for nearest-neighbor regression in the notation used by RAG-PL \citep{Jiang2019nonasymptoticuniform}. Its use as an action-specific matching estimator follows the setup of \citet{Abadie2006largesample}.

\begin{theorem}[Estimation error of the expected outcome from nearest-neighbor matching]
\label{thm:nn-expected-outcome}
Fix $a\in\cb{0,1}$. Suppose that the support of $H_{a,j}$ has finite diameter $D_H$ and that the following conditions hold.
\begin{enumerate}
    \item The function $m_a$ is $\beta$-Hölder for some $\beta\in(0,1]$ and $L_a>0$:
    \[
    |m_a(h)-m_a(h')|
    \leq
    L_a\lVert h-h'\rVert^\beta.
    \]
    \item For some $c_a>0$ and $r_0>0$, every query embedding $h$ and every $0<r\leq r_0$ satisfy
    \[
    Q_a(B(h,r))\geq c_ar^{d_\varphi},
    \]
    where $B(h,r)\coloneqq\cb{h'\in\mathbb R^{d_\varphi}:\lVert h'-h\rVert\leq r}$.
    \item Conditional on the embeddings, the errors $\xi_{a,j}$ are independent, mean zero, and $\sigma_a^2$-sub-Gaussian for some $\sigma_a>0$.
\end{enumerate}
Let
\[
r_{a,k}
\coloneqq
\left(\frac{2k}{c_aN_a}\right)^{1/d_\varphi},
\qquad
b_{\varphi,a}
\coloneqq
\sup_{x\in\calX}|m_a(\varphi(x))-f_0(a,x)|,
\]
and suppose that $r_{a,k}\leq r_0$. Then, for every $x\in\calX$ and $u>0$, we have
\[
\Pr\left(
|\widetilde f_k(a,x)-f_0(a,x)|
\geq
b_{\varphi,a}+L_ar_{a,k}^\beta+u
\right)
\leq
2\exp\left(-\frac{ku^2}{2\sigma_a^2}\right)
+
\exp(-k/4).
\]
Moreover, we have
\[
\bbE\sqb{(\widetilde f_k(a,X)-f_0(a,X))^2}
\leq
C\left(
 b_{\varphi,a}^2
 +
 \left(\frac{k}{N_a}\right)^{2\beta/d_\varphi}
 +
 \frac{1}{k}
 +
 \exp(-k/4)
\right).
\]
\end{theorem}
The constant $C$ depends only on the constants in the assumptions.

The term $b_{\varphi,a}$ is zero when the embedding retains all information needed for the conditional expected outcome, so $f_0(a,x)=m_a(\varphi(x))$. Otherwise, it records the difference between conditioning on $X$ and conditioning on its embedding. The local-mass condition requires enough observations under action $a$ near each query. It complements the positivity condition used for causal identification.

The displayed estimator is a benchmark for the expected-outcome RAG. Suppose that the RAG output satisfies
\[
\max_{a\in\cb{0,1}}\sup_{x\in\calX}
|\widehat f(a,x)-\widetilde f_k(a,x)|
\leq
\eta_n.
\]
Suppose that the margin condition \eqref{eq:binary-margin} holds. For the deviation-based bound below, also suppose that $|\tau_0(X)|\leq B_\tau$ holds almost surely. Let $N_{\min}=\min\cb{N_0,N_1}$ and $b_\varphi=b_{\varphi,0}+b_{\varphi,1}$. Combining Theorem~\ref{thm:nn-expected-outcome} with Theorem~\ref{thm:binary-regret} gives
\begin{align}
R(\widehat a)
&\leq
C\left(
 b_\varphi
 +
 \left(\frac{k}{N_{\min}}\right)^{\beta/d_\varphi}
 +
 k^{-1/2}
 +
 2\eta_n
\right)^{1+\kappa}
+
2B_\tau\exp(-k/4),
\label{eq:nn-sharp-regret}
\end{align}
provided that the term in parentheses is sufficiently small relative to $t_0$. The MSE statement gives the alternative bound
\begin{equation}
\label{eq:nn-mse-regret}
R(\widehat a)
\leq
C\left(
 b_\varphi^2
 +
 \left(\frac{k}{N_{\min}}\right)^{2\beta/d_\varphi}
 +
 \frac{1}{k}
 +
 \eta_n^2
\right)^{(1+\kappa)/(2+\kappa)}.
\end{equation}
Choosing $k\asymp N_{\min}^{2\beta/(2\beta+d_\varphi)}$ balances the nearest-neighbor bias and variance. In \eqref{eq:nn-sharp-regret}, this gives the root prediction rate $N_{\min}^{-\beta/(2\beta+d_\varphi)}$ before the margin exponent is applied. In \eqref{eq:nn-mse-regret}, the term $\exp(-k/4)$ is absorbed into $1/k$.

For a general document corpus without observed outcomes, Theorem~\ref{thm:nn-expected-outcome} does not apply directly. In that case, the theorem describes the matching calculation that the RAG output is meant to approximate, and a regret rate requires an assumption or an evaluation of the difference represented by $\eta_n$.

\subsection{Matching Weights and Propensity Scores}
Similarity is used to select comparable evidence, not to choose the action that appears most often in nearby cases. Let $\mathcal N_k^H(h)$ denote the indices of the $k$ nearest observations
to $h$ in the pooled embedding database. If
\[
\widehat e_k^H(a\mid h)
=
\frac{1}{k}\sum_{j\in\mathcal N_k^H(h)}
\mathbbm{1}\sqb{A_j=a},
\]
then $\widehat e_k^H(a\mid h)$ estimates
$\Pr(A=a\mid H=h)$ under standard nearest-neighbor conditions.
At $h=\varphi(x)$, this quantity equals $\Pr(A=a\mid X=x)$ only under
an additional condition such as $A\perp X\mid H$. Maximizing this quantity selects the most common past action near $x$, not the action with the largest conditional expected outcome. The joint distribution of $(X,A)$ can remain unchanged while the conditional outcomes are changed so that either action is optimal. Thus, observations of contexts and actions alone do not determine the optimal policy.

Nearest-neighbor matching also has a weighting representation. Let $X_1^{\mathrm{tar}},\dots,X_m^{\mathrm{tar}}$ be target contexts and define
\[
\widetilde f_k(a,X_i^{\mathrm{tar}})
=
\frac{1}{k}\sum^{N_a}_{j=1}
\mathbbm{1}\sqb{j\in\mathcal N_{k,a}(\varphi(X_i^{\mathrm{tar}}))}Y_{a,j}.
\]
Define
\[
K_{a,j}
\coloneqq
\sum^m_{i=1}\mathbbm{1}\sqb{j\in\mathcal N_{k,a}(\varphi(X_i^{\mathrm{tar}}))}.
\]
Then, changing the order of summation gives
\[
\frac{1}{m}\sum^m_{i=1}\widetilde f_k(a,X_i^{\mathrm{tar}})
=
\sum^{N_a}_{j=1}\frac{K_{a,j}}{mk}Y_{a,j}.
\]
The match count therefore acts as a weight when local predictions are averaged over target contexts. For the propensity-score interpretation, suppose that the action-specific
databases are sampled from the same observational population, so
$Q_a=P_{H\mid A=a}$, and that the target embeddings follow $P_H$.
Bayes' rule then gives
\[
\frac{p_H(h)}{p_{H\mid A=a}(h)}
=
\frac{\Pr(A=a)}{\Pr(A=a\mid H=h)}.
\]
The density ratio is therefore proportional to an inverse propensity score
defined conditional on the embedding.
The density ratio is proportional to an inverse propensity score. \citet{Lin2023estimationbased} studies the density-ratio limit of nearest-neighbor match counts, and \citet{Kato2025nearestneighbor} relates nearest-neighbor matching to least-squares density-ratio estimation and Riesz regression. These results connect matching with IPW and AIPW representations when outcomes are available. They do not make the propensity score an expected outcome.

\section{Discussion}
\label{sec:discussion}

\subsection{Identification}
The regret analysis is written in terms of $f_0(a,x)=\bbE\sqb{Y\mid A=a,X=x}$. A causal interpretation requires consistency, conditional exchangeability, and positivity. Under these conditions, it holds that
\[
\bbE\sqb{Y\mid A=a,X=x}
=
\bbE\sqb{Y(a)\mid X=x}.
\]
The information contained in $X$ is therefore part of the identification argument. It is natural to write $X=(Q,W)$, where $Q$ is the query and $W$ contains variables observed before the action, such as age, gender, occupation, and prior decisions. If a variable affects both past action selection and the outcome, omitting it can prevent the conditional comparison from having a causal interpretation. Including more variables does not establish conditional exchangeability by itself, and post-action variables should not be used as controls.

Vector search is usually performed on an embedding $H=\varphi(X)$. Matching on $H$ is sufficient only when the representation retains the information required for adjustment or for the conditional expected outcomes. The GenAI-Powered Inference framework combines structured covariates with low-dimensional features extracted from unstructured data \citep{Imai2025genai}. This suggests keeping personal attributes explicit while representing the query text through an embedding. Embedding-powered BISG provides an example in which embeddings are used to estimate a missing personal attribute probabilistically \citep{Dasanaike2026usingembedding}. Such an estimate is a probabilistic proxy rather than an observed covariate, so its estimation error and possible distribution shift remain relevant. The variables used for adjustment can also be richer than the variables allowed in the final policy. In that case, expected outcomes are first adjusted using the larger information set and then averaged over variables excluded from the policy.

Instrumental variables provide another identification strategy when conditional exchangeability fails in the data used to estimate expected outcomes. They enter through the source-data model, not by treating the RAG output as an instrument. Existing studies consider optimal treatment regimes under additional IV assumptions \citep{Cui2021asemiparametric,Qiu2021optimalindividualized} and decision rules under partial identification \citep{Pu2021estimatingoptimal}. These problems have different target values. Applying RAG-PL to them would require the expected outcome supplied to the RAG system to be defined for the corresponding target.

\subsection{Expected Outcome and Policy-Value Formulations}
When $\widehat f(a,x)$ is available and the policy class contains all measurable selectors, maximizing $\widehat f(a,x)$ at each context is equivalent to maximizing $\bbE\sqb{\widehat f(s(X),X)}$ over the policy class. The same equivalence holds for an empirical average over target contexts. A restricted policy class can introduce an approximation loss, but the criterion is still constructed from the same expected-outcome estimate. The relation between EWM and least-squares CATE estimation for a reparameterized binary policy class is studied by \citet{Kato2025bridginggap}.

A different analysis is needed when the policy value is constructed from observed outcomes rather than from $\widehat f$ alone. IPW uses weighted outcomes, and AIPW combines an expected-outcome estimate with an outcome-residual correction \citep{Kitagawa2018whoshould,Athey2021policylearning}. In the main setting of this study, the analyst starts from a trained RAG system and may not observe unit-level outcomes and assignment probabilities. The fact that outcome-related text may have appeared during pretraining does not provide the data required to construct IPW or AIPW. If a retrieval database contains unit-level observations of $(X,A,Y)$, IPW and AIPW can instead be constructed directly from those data. For one-step RAG-PL, the internal criterion is not observed, so we evaluate the returned policy without assigning it to one of these formulations.

\subsection{Margin Conditions}
The margin condition is separate from identification. Identification determines whether the contrast has a causal interpretation, while the margin condition controls how often estimation error changes the preferred action. In the binary case, it bounds the probability that $|\tau_0(X)|$ is close to zero. When few contexts lie near this boundary, the regret decreases faster than the expected-outcome error.

The margin condition depends on the information used to define the conditional expected outcomes. Omitting a relevant attribute can average positive and negative effects into a contrast near zero, while adding an attribute can also reveal finer effects that lie near zero. Replacing $X$ by an embedding can change the margin for the same reason. For two-step RAG-PL, the margin within $g(x)$ is also distinct from the quality of the generated set. A large gap among poor candidates does not remove $R_{\mathrm{gen}}(g)$.

\subsection{Action Generation in Large Action Spaces}
When $\calA$ is finite and given, $R_{\mathrm{gen}}(g)=0$. For a large or uncountable action space, the generated set determines which actions can be compared. Exact inclusion of $a^*(x)$ is stronger than necessary; Proposition~\ref{prop:candidate-bounds} uses an $\varepsilon$-optimal action, a coverage probability, or distance in the action space. These conditions describe the candidate set through the best conditional expected outcome it contains.

Increasing the number of candidates can improve the chance of including a good action, but more expected outcomes must then be estimated or ranked. Equation~\eqref{eq:two-step-decomposition} keeps these effects separate. A useful candidate-set size depends on both the quality of generation and the accuracy of the second step.

\section{Simulation Studies}
\label{sec:simulation}
We evaluate whether action-specific retrieval of cases with similar pre-action information improves policy choice. Each data-generating process first produces a structured pre-action state together with the conditional expected outcome under every available action. We convert the structured states into complete English queries and case reports before executing the notebooks. The RAG procedures observe only the resulting texts, while the numerical state is retained to calculate the optimal action and regret. Because each numerical level is mapped deterministically to a verbal category, the text retains the full discrete state information. Appendix~\ref{app:simulation-designs} gives the full specifications.

\paragraph{Data-generating processes.}
DGP~1 and DGP~2 use two actions. Action~0 is temporary workload reduction and action~1 is weekly individual coaching. Let $X=(X_1,\ldots,X_6)$ describe recent performance; workload; task complexity; experience; schedule flexibility; and team support. Each component is independently and uniformly distributed on $\cb{-2,-1,0,1,2}$. DGP~1 has a linear conditional effect and observational action assignment. DGP~2 has randomized action assignment and a smooth nonlinear conditional effect whose population mean is zero. The preferred action therefore varies with the current state even though neither constant action is favored on average. DGP~3 contains 24 interventions formed from four intervention types with three intensity levels and two durations. Its conditional expected outcome rewards agreement between the employee's stated needs and the intervention profile.

\paragraph{Methods and implementation.}
The main results use GPT-5.4~mini with temperature zero. Retrieval vectors are obtained from \texttt{text-embedding-3-small}. The embedding input contains only a fixed-order description of the pre-action information. It excludes the historical action and the observed outcome. It also excludes the decision question. Cosine distance is used after normalization.

In DGP~1 and DGP~2, One-step RAG-PL and Two-step RAG-PL receive the same current query and the same six nearest reports under each action. The twelve reports are presented in alternating action order. The prompt also gives the exact mean of the six displayed scores under each action. Providing these means allows the model to use exact local sample averages, so the experiment does not test whether it can calculate those averages from the reports. One-step RAG-PL returns an action directly. Two-step RAG-PL returns the two action-specific conditional expected outcomes, and the experiment code selects the larger estimate. RAG without covariates receives neither the current profile nor historical pre-action profiles. It uses six randomly sampled action-and-outcome reports under each action and returns one population-level decision for the retrieval database. That decision is applied to all ten test queries in the repetition. RAG without covariates is therefore a population-level baseline that differs from RAG-PL in both the current-state information supplied and the retrieval rule.
In DGP~3, both RAG-PL procedures receive the same twelve nearest reports and the complete catalog of 24 actions. One-step RAG-PL returns one action directly. Two-step RAG-PL generates five distinct candidate actions. It retrieves five reports observed under each candidate and estimates the five candidate-specific conditional expected outcomes. It then selects the candidate with the largest estimate.

\paragraph{Evaluation.}
Each DGP is repeated 20 times. Each repetition contains ten test queries. DGP~1 and DGP~2 use 500 historical reports, while DGP~3 uses 600. For each repetition, we first average regret and policy value over its ten queries. Table~\ref{tab:simulation-main} reports the mean of these repetition-level quantities. The regret standard error is calculated across the 20 repetitions. The optimal-action probability counts any action attaining the largest true conditional expected outcome as optimal. Figures~\ref{fig:dgp1-regret}--\ref{fig:dgp3-regret} show the repetition-level mean regret.

\begin{table}[t]
\centering
\caption{Main simulation results with GPT-5.4 mini}
\label{tab:simulation-main}
\small
\begin{tabular}{lrrrr}
\toprule
Method & Regret & Standard error & Policy value & Optimal-action probability \\
\midrule
\multicolumn{5}{l}{\textbf{DGP 1}} \\
One-step RAG-PL & \textbf{3.3700} & 0.4200 & \textbf{66.6925} & 0.5500 \\
Two-step RAG-PL & \textbf{3.3700} & 0.4266 & \textbf{66.6925} & \textbf{0.5550} \\
RAG without covariates & 4.5450 & 0.3943 & 65.5175 & 0.4100 \\
\midrule
\multicolumn{5}{l}{\textbf{DGP 2}} \\
One-step RAG-PL & \textbf{1.5676} & 0.1564 & \textbf{65.7886} & \textbf{0.6400} \\
Two-step RAG-PL & 1.6071 & 0.1746 & 65.7491 & \textbf{0.6400} \\
RAG without covariates & 2.3638 & 0.2192 & 64.9924 & 0.5100 \\
\midrule
\multicolumn{5}{l}{\textbf{DGP 3}} \\
One-step RAG-PL & 1.4985 & 0.1134 & 65.1192 & 0.0950 \\
Two-step RAG-PL & \textbf{1.2154} & 0.0837 & \textbf{65.4023} & \textbf{0.1400} \\
\bottomrule
\end{tabular}
\end{table}

\begin{figure}[t]
\centering
\includegraphics[draft=false,width=0.82\linewidth]{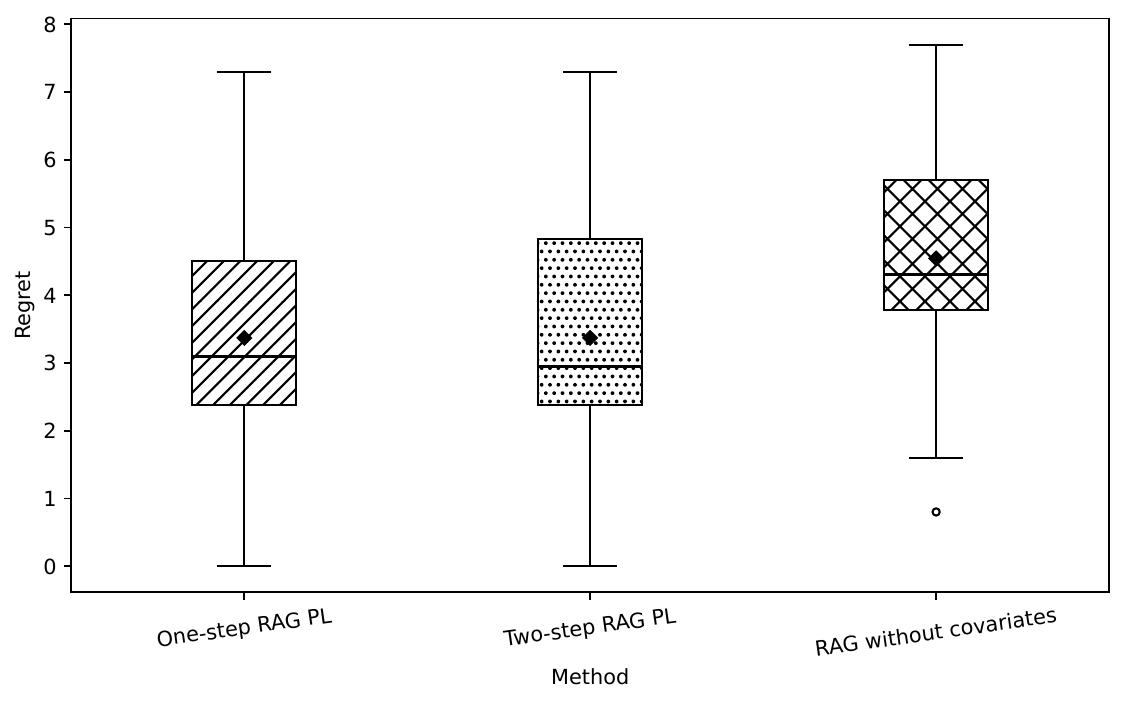}
\caption{Repetition-level mean regret in DGP~1. The diamond denotes the mean across repetitions.}
\label{fig:dgp1-regret}
\end{figure}

\begin{figure}[t]
\centering
\includegraphics[draft=false,width=0.82\linewidth]{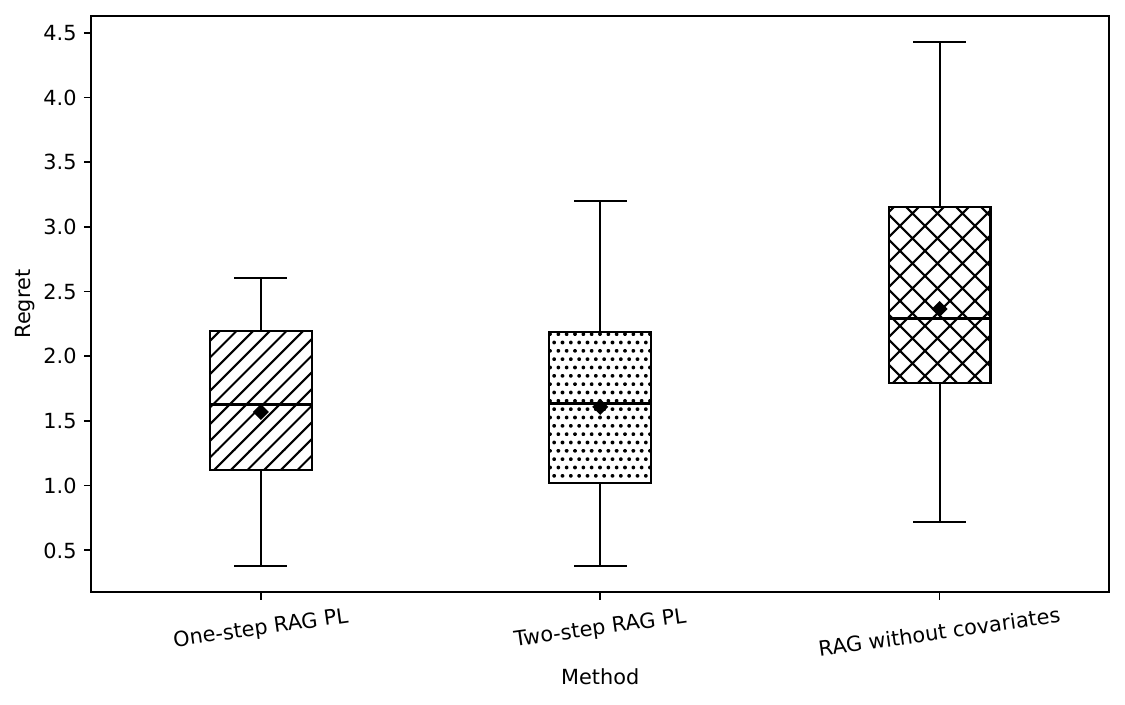}
\caption{Repetition-level mean regret in DGP~2. The diamond denotes the mean across repetitions.}
\label{fig:dgp2-regret}
\end{figure}

\begin{figure}[t]
\centering
\includegraphics[draft=false,width=0.76\linewidth]{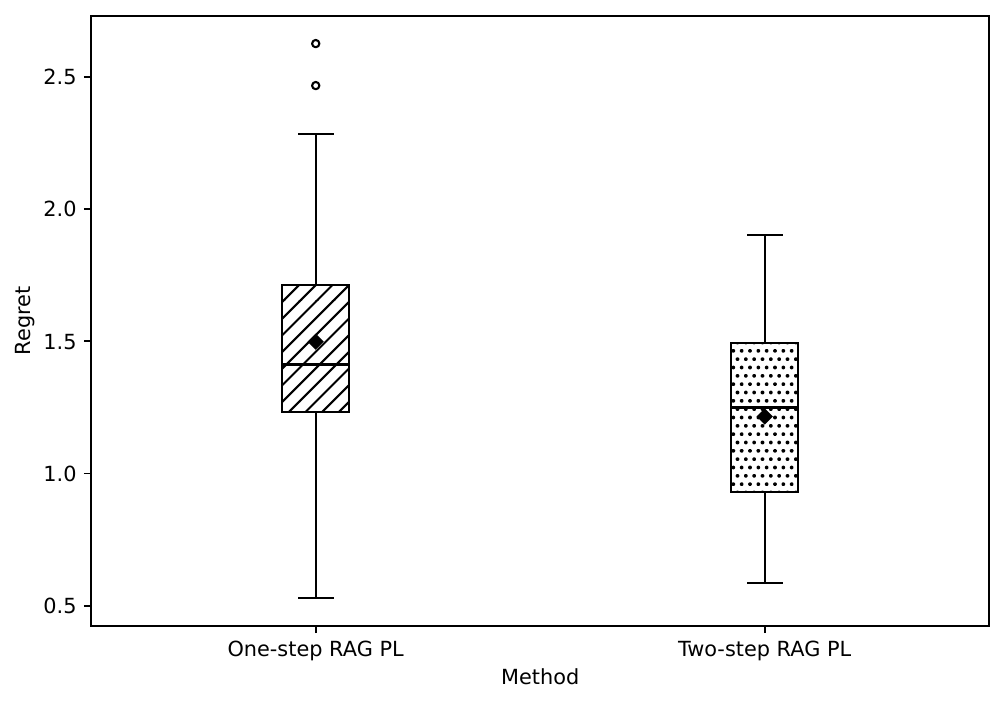}
\caption{Repetition-level mean regret in DGP~3. The diamond denotes the mean across repetitions.}
\label{fig:dgp3-regret}
\end{figure}

\paragraph{Results.}
In DGP~1, One-step RAG-PL and Two-step RAG-PL both reduce mean regret from $4.5450$ under RAG without covariates to $3.3700$. The paired mean-regret reduction is $1.1750$ for each procedure, with standard errors of $0.3562$ for One-step RAG-PL and $0.3635$ for Two-step RAG-PL. The optimal-action probability rises from $0.4100$ under RAG without covariates to $0.5500$ under One-step RAG-PL and $0.5550$ under Two-step RAG-PL.

DGP~2 isolates personalization because historical actions are randomized. One-step RAG-PL has mean regret $1.5676$ and Two-step RAG-PL has mean regret $1.6071$, compared with $2.3638$ for RAG without covariates. The paired reductions relative to RAG without covariates are $0.7963$ with standard error $0.2257$ and $0.7567$ with standard error $0.2473$, respectively. Both RAG-PL procedures select an optimal action with probability $0.6400$, compared with $0.5100$ under RAG without covariates.

In DGP~3, Two-step RAG-PL lowers mean regret from $1.4985$ to $1.2154$. The paired difference, Two-step RAG-PL minus One-step RAG-PL, is $-0.2831$ with standard error $0.1431$. The regret decomposition for Two-step RAG-PL gives candidate-set regret $0.4822$ and within-candidate regret $0.7332$. Thus, both candidate generation and selection among the generated candidates contribute to the remaining regret. The generated set contains an exactly optimal action with probability $0.4450$ and an action within $0.5$ of the optimum with probability $0.6700$. Additional estimation and robustness results are reported in Appendix~\ref{app:simulation-results}.

\section{Conclusion}
We proposed one-step and two-step RAG-based policy learning methods. We interpreted the vector-search component of RAG as nearest-neighbor matching and related this matching procedure to propensity-score weighting, density-ratio estimation, and Riesz regression. This interpretation connects RAG-based decision-making to the existing literature on policy learning. We then derived regret upper bounds for the proposed methods by combining policy-learning analysis with nonparametric prediction-error bounds for nearest-neighbor estimators and transformers. In the simulation studies, both RAG-PL methods had lower mean regret than RAG without covariates in the binary settings, and the two-step method had lower mean regret than the one-step method in the 24-action setting.

\bibliography{arXiv2.bbl}

\bibliographystyle{tmlr}

\onecolumn

\appendix

\section{Proofs for the Regret Analysis}
\label{app:regret-proofs}

\subsection{Regret decomposition and candidate-set bounds}
For every $x\in\calX$, adding and subtracting $f_0(a_g^*(x),x)$ gives
\begin{align*}
&f_0(a^*(x),x)-f_0(\widehat a_{\mathrm{two}}(x),x)\\
&=f_0(a^*(x),x)-f_0(a_g^*(x),x)
+f_0(a_g^*(x),x)-f_0(\widehat a_{\mathrm{two}}(x),x).
\end{align*}
Taking expectations proves \eqref{eq:two-step-decomposition}.

\begin{proof}[Proof of Proposition~\ref{prop:candidate-bounds}]
For the first statement, suppose that $g(x)\cap\calA_\varepsilon^*(x)\neq\varnothing$. Then, there is an action $a\in g(x)$ such that
\[
f_0(a^*(x),x)-f_0(a,x)\leq\varepsilon.
\]
Since $a_g^*(x)$ maximizes $f_0(a,x)$ over $g(x)$, it holds that $f_0(a_g^*(x),x)\geq f_0(a,x)$. Therefore,
\[
f_0(a^*(x),x)-f_0(a_g^*(x),x)\leq\varepsilon.
\]
Taking expectations proves the first statement.

For the second statement, define
\[
E_\varepsilon
\coloneqq
\cb{g(X)\cap\calA_\varepsilon^*(X)\neq\varnothing}.
\]
The first statement gives a pointwise bound of $\varepsilon$ on $E_\varepsilon$, while the bounded-difference assumption gives a bound of $B_f$ on $E_\varepsilon^c$. Hence,
\begin{align*}
R_{\mathrm{gen}}(g)
&\leq
\varepsilon\Pr(E_\varepsilon)+B_f\Pr(E_\varepsilon^c)\\
&\leq
\varepsilon+B_f\delta.
\end{align*}

For the third statement, fix $x$. Because $g(x)$ is finite, there is an action $\widetilde a_g(x)\in g(x)$ such that
\[
d_\calA(a^*(x),\widetilde a_g(x))
=
d_\calA(a^*(x),g(x)).
\]
Since $a_g^*(x)$ maximizes $f_0$ over $g(x)$,
\begin{align*}
f_0(a^*(x),x)-f_0(a_g^*(x),x)
&\leq
f_0(a^*(x),x)-f_0(\widetilde a_g(x),x)\\
&\leq
L_\calA d_\calA(a^*(x),g(x))^{\beta_\calA}.
\end{align*}
Taking expectations completes the proof.
\end{proof}

\subsection{Expected-outcome estimation}
\begin{proof}[Proof of Theorem~\ref{thm:expected-outcome-regret}]
Fix $x\in\calX$. The definition of $\widehat a_{\mathrm{two}}(x)$ gives
\[
\widehat f(a_g^*(x),x)
-
\widehat f(\widehat a_{\mathrm{two}}(x),x)
\leq0.
\]
Therefore,
\begin{align*}
&f_0(a_g^*(x),x)-f_0(\widehat a_{\mathrm{two}}(x),x)\\
&=f_0(a_g^*(x),x)-\widehat f(a_g^*(x),x)\\
&\quad+\widehat f(a_g^*(x),x)-\widehat f(\widehat a_{\mathrm{two}}(x),x)\\
&\quad+\widehat f(\widehat a_{\mathrm{two}}(x),x)-f_0(\widehat a_{\mathrm{two}}(x),x)\\
&\leq2\delta_f(x;g).
\end{align*}
Taking expectations and using \eqref{eq:two-step-decomposition} proves the first two statements.

Suppose now that $\delta_f(X;g)\leq r_f$. If the selected action is not optimal in $g(X)$, then
\[
0<\Delta_g(X)
\leq
f_0(a_g^*(X),X)-f_0(\widehat a_{\mathrm{two}}(X),X)
\leq2r_f.
\]
It follows that
\begin{align*}
R_{\mathrm{choice}}(\widehat a_{\mathrm{two}};g)
&\leq
2r_f\Pr\p{0<\Delta_g(X)\leq2r_f}\\
&\leq
C_M(2r_f)^{1+\kappa}.
\end{align*}
\end{proof}

To derive \eqref{eq:high-prob-choice}, let $E_f=\cb{\delta_f(X;g)\leq r_f}$. On $E_f$, the preceding margin argument applies. On $E_f^c$, the value difference is at most $B_f$. Hence,
\begin{align*}
R_{\mathrm{choice}}(\widehat a_{\mathrm{two}};g)
&\leq
C_M(2r_f)^{1+\kappa}+B_f\Pr(E_f^c)\\
&\leq
C_M(2r_f)^{1+\kappa}+B_f\eta_f.
\end{align*}
The finite-candidate tail bound follows from
\begin{align*}
&\Pr\left(
\max_{a\in g(X)}|\widehat f(a,X)-f_0(a,X)|>b_n+u
\mathrel{\Big|}
X,g(X)
\right)\\
&\leq
\sum_{a\in g(X)}
\Pr\left(
|\widehat f(a,X)-f_0(a,X)|>b_n+u
\mathrel{\Big|}
X,g(X)
\right)\\
&\leq
c_1M\exp(-c_2a_nu^2).
\end{align*}

\subsection{Relation between the one-step and two-step methods}
Suppose that $\widehat a_{\mathrm{one}}(x)\in g(x)$. Then,
\[
f_0(a_g^*(x),x)
\geq
f_0(\widehat a_{\mathrm{one}}(x),x),
\]
so
\[
f_0(a^*(x),x)-f_0(a_g^*(x),x)
\leq
f_0(a^*(x),x)-f_0(\widehat a_{\mathrm{one}}(x),x).
\]
Taking expectations gives $R_{\mathrm{gen}}(g)\leq R(\widehat a_{\mathrm{one}})$. The remaining inequalities in
Section~\ref{sec:regret-analysis} follow from \eqref{eq:two-step-decomposition} and Theorem~\ref{thm:expected-outcome-regret}.

\subsection{Proof of Theorem~\ref{thm:binary-regret}}
The identity \eqref{eq:binary-regret-identity} follows pointwise. If $\widehat a(x)=a^*(x)$, the loss is zero. Otherwise, the selected action has the lower conditional expected outcome, and the difference is $|\tau_0(x)|$.

For the first statement, suppose that $\widehat a(x)\neq a^*(x)$ and $\tau_0(x)\neq0$. Then, $\tau_0(x)\widehat\tau(x)\leq0$, which implies
\[
0<|\tau_0(x)|
\leq
|\widehat\tau(x)-\tau_0(x)|
\leq
r_\tau.
\]
Using \eqref{eq:binary-regret-identity} and the margin condition,
\begin{align*}
R(\widehat a)
&\leq
\bbE\sqb{|\tau_0(X)|\mathbbm{1}\sqb{0<|\tau_0(X)|\leq r_\tau}}\\
&\leq
r_\tau\Pr\p{0<|\tau_0(X)|\leq r_\tau}\\
&\leq
C_Mr_\tau^{1+\kappa}.
\end{align*}
If each expected-outcome error is at most $r_f$, then
\[
|\widehat\tau(x)-\tau_0(x)|
\leq
|\widehat f(1,x)-f_0(1,x)|
+
|\widehat f(0,x)-f_0(0,x)|
\leq2r_f.
\]

For the second statement, put $T=|\tau_0(X)|$. On a sign error with $T>0$, it holds that $|\widehat\tau(X)-\tau_0(X)|\geq T$. The contribution from $0<T\leq2b_n$ is at most
\[
2b_n\Pr\p{0<T\leq2b_n}
\leq
C_M(2b_n)^{1+\kappa}.
\]
On $T>2b_n$, the assumed deviation inequality gives, conditional on $X$,
\[
\Pr\p{\widehat a(X)\neq a^*(X)\mid X}
\leq
c_1\exp(-c_2a_nT^2/4)+\delta_n.
\]
Let $s=a_n^{-1/2}$ and $c=c_2/4$. The region $0<T\leq s$ contributes at most $C_Ms^{1+\kappa}$. For $j\geq0$, define
\[
S_j
\coloneqq
\cb{2^js<T\leq2^{j+1}s}.
\]
Whenever $2^{j+1}s\leq t_0$,
\begin{align*}
\bbE\sqb{T\exp(-ca_nT^2)\mathbbm{1}\sqb{S_j}}
&\leq
2^{j+1}s\exp(-c4^j)\Pr\p{0<T\leq2^{j+1}s}\\
&\leq
C_Ms^{1+\kappa}2^{(j+1)(1+\kappa)}\exp(-c4^j).
\end{align*}
The series over $j$ is finite. On $T>t_0$, boundedness gives a term of order $\exp(-ca_nt_0^2)$, which is bounded by a constant multiple of $s^{1+\kappa}$ under the stated condition. The contribution of $\delta_n$ is at most $B_\tau\delta_n$. Therefore,
\[
R(\widehat a)
\leq
C_1b_n^{1+\kappa}+C_2a_n^{-(1+\kappa)/2}+B_\tau\delta_n,
\]
and the displayed result follows.

For the third statement, let $E(X)=\widehat\tau(X)-\tau_0(X)$ and choose $t\in(0,t_0]$. By the sign-error containment,
\begin{align*}
R(\widehat a)
&\leq
\bbE\sqb{|\tau_0(X)|\mathbbm{1}\sqb{0<|\tau_0(X)|\leq t}}\\
&\quad+
\bbE\sqb{|\tau_0(X)|\mathbbm{1}\sqb{\widehat a(X)\neq a^*(X),\ |\tau_0(X)|>t}}.
\end{align*}
The first term is at most $C_Mt^{1+\kappa}$. On the event in the second term, $|E(X)|\geq|\tau_0(X)|>t$, so
\[
|\tau_0(X)|
\leq
\frac{E(X)^2}{t}.
\]
Consequently,
\[
R(\widehat a)
\leq
C_Mt^{1+\kappa}+\frac{q_n}{t}.
\]
Taking $t=q_n^{1/(2+\kappa)}$ proves the result.

\section{Proofs for Nearest-Neighbor Matching}
\label{app:nn-proofs}

\subsection{Nearest-neighbor radius}
Fix an action $a$ and a query embedding $h$. Let
\[
R_{a,k}(h)
\coloneqq
\max_{j\in\mathcal N_{k,a}(h)}\lVert H_{a,j}-h\rVert.
\]
The event $R_{a,k}(h)>r_{a,k}$ occurs only if fewer than $k$ database points fall in $B(h,r_{a,k})$. The number of points in that ball is binomial with mean at least
\[
N_aQ_a(B(h,r_{a,k}))
\geq
N_ac_ar_{a,k}^{d_\varphi}
=
2k.
\]
The multiplicative Chernoff inequality therefore gives
\begin{equation}
\label{eq:nn-radius-tail}
\Pr\p{R_{a,k}(h)>r_{a,k}}
\leq
\exp(-k/4).
\end{equation}
Since the support has diameter $D_H$,
\[
\bbE\sqb{R_{a,k}(h)^{2\beta}}
\leq
r_{a,k}^{2\beta}+D_H^{2\beta}\exp(-k/4).
\]

\subsection{Proof of Theorem~\ref{thm:nn-expected-outcome}}
On the event $R_{a,k}(h)\leq r_{a,k}$, decompose
\begin{align*}
\widetilde f_k(a,x)-f_0(a,x)
={}&
\frac{1}{k}\sum_{j\in\mathcal N_{k,a}(h)}\cb{m_a(H_{a,j})-m_a(h)}\\
&+
\frac{1}{k}\sum_{j\in\mathcal N_{k,a}(h)}\xi_{a,j}\\
&+
m_a(h)-f_0(a,x).
\end{align*}
The absolute values of the first and third terms are bounded by $L_ar_{a,k}^\beta$ and $b_{\varphi,a}$, respectively. Conditional on the embeddings, the middle term is mean zero and sub-Gaussian with variance proxy $\sigma_a^2/k$. Hence,
\[
\Pr\left(
\left|\frac{1}{k}\sum_{j\in\mathcal N_{k,a}(h)}\xi_{a,j}\right|\geq u
\mathrel{\Big|}
\cb{H_{a,j}}_{j=1}^{N_a}
\right)
\leq
2\exp\left(-\frac{ku^2}{2\sigma_a^2}\right).
\]
Combining this inequality with \eqref{eq:nn-radius-tail} proves the pointwise deviation bound.

For the MSE bound, condition on the embeddings. The conditional variance is at most $\sigma_a^2/k$. The squared conditional bias is bounded by a constant multiple of
\[
b_{\varphi,a}^2+L_a^2R_{a,k}(h)^{2\beta}.
\]
Taking expectations and using the radius moment bound gives
\[
\bbE\sqb{(\widetilde f_k(a,x)-f_0(a,x))^2}
\leq
C\left(
 b_{\varphi,a}^2
 +
 \left(\frac{k}{N_a}\right)^{2\beta/d_\varphi}
 +
 \frac{1}{k}
 +
 \exp(-k/4)
\right).
\]
Integrating over $X$ proves the stated MSE result.

\subsection{Derivation of the regret bounds}
Let
\[
\widetilde\tau_k(x)
=
\widetilde f_k(1,x)-\widetilde f_k(0,x).
\]
A union bound and Theorem~\ref{thm:nn-expected-outcome} give, for every $x$ and $u>0$,
\begin{align*}
&\Pr\left(
|\widetilde\tau_k(x)-\tau_0(x)|
\geq
b_\varphi+L_0r_{0,k}^\beta+L_1r_{1,k}^\beta+u
\right)\\
&\leq
4\exp\left(-\frac{ku^2}{8\sigma_{\max}^2}\right)
+2\exp(-k/4),
\end{align*}
where $\sigma_{\max}=\max(\sigma_0,\sigma_1)$. If $\widehat f$ differs from $\widetilde f_k$ by at most $\eta_n$ for each action, then
\[
|\widehat\tau(x)-\widetilde\tau_k(x)|\leq2\eta_n.
\]
The second statement of Theorem~\ref{thm:binary-regret} applies with a deterministic term of order
\[
b_\varphi
+
\left(\frac{k}{N_{\min}}\right)^{\beta/d_\varphi}
+
2\eta_n,
\]
a concentration scale $k^{-1/2}$, and a failure probability $2\exp(-k/4)$. This proves \eqref{eq:nn-sharp-regret}.

For the MSE bound, the inequality $(u+v)^2\leq2u^2+2v^2$ gives
\[
\bbE\sqb{(\widehat\tau(X)-\tau_0(X))^2}
\leq
C\left(
 b_\varphi^2
 +
 \left(\frac{k}{N_{\min}}\right)^{2\beta/d_\varphi}
 +
 \frac{1}{k}
 +
 \eta_n^2
 +
 \exp(-k/4)
\right).
\]
The third statement of Theorem~\ref{thm:binary-regret} proves \eqref{eq:nn-mse-regret}. Balancing $(k/N_{\min})^{\beta/d_\varphi}$ and $k^{-1/2}$ gives $k\asymp N_{\min}^{2\beta/(2\beta+d_\varphi)}$.

\section{Causal Identification with Covariates and Embeddings}
\label{app:identification}
This section records the conditions under which the observational conditional mean can be interpreted as a potential-outcome conditional mean. We use a discrete action space because that is the setting needed for the matching discussion.

\begin{proposition}[Identification with observed covariates]
\label{prop:identification-x}
Suppose that consistency holds, so $Y=Y(A)$. Suppose also that
\[
\cb{Y(a):a\in\calA}
\perp
A
\mid
X
\]
and that $\Pr(A=a\mid X=x)>0$ on the covariate support of interest. Then,
\[
f_0(a,x)
=
\bbE\sqb{Y(a)\mid X=x}.
\]
\end{proposition}

\begin{proof}
By consistency, $Y=Y(a)$ on the event $A=a$. Conditional exchangeability gives
\begin{align*}
f_0(a,x)
&=\bbE\sqb{Y\mid A=a,X=x}\\
&=\bbE\sqb{Y(a)\mid A=a,X=x}\\
&=\bbE\sqb{Y(a)\mid X=x}.
\end{align*}
Positivity ensures that the conditional mean given $A=a$ is defined on the target support.
\end{proof}

Let $H=\varphi(X)$. The standard balancing-score argument gives the following result \citep{Rosenbaum1983centralrole}.

\begin{proposition}[Identification with a balancing representation]
\label{prop:balancing-representation}
Suppose that the conditions of Proposition~\ref{prop:identification-x} hold and that $A\perp X\mid H$. Then, $Y(a)\perp A\mid H$ for every action $a$. If $\Pr(A=a\mid H=h)>0$, then
\[
\bbE\sqb{Y(a)\mid H=h}
=
\bbE\sqb{Y\mid A=a,H=h}.
\]
\end{proposition}

\begin{proof}
For every bounded measurable function $u$, conditional exchangeability given $X$ implies
\[
\bbE\sqb{u(Y(a))\mid A,X,H}
=
\bbE\sqb{u(Y(a))\mid X}.
\]
Taking the conditional expectation given $(A,H)$ and using $A\perp X\mid H$ gives
\begin{align*}
\bbE\sqb{u(Y(a))\mid A,H}
&=\bbE\sqb{\bbE\sqb{u(Y(a))\mid X}\mid A,H}\\
&=\bbE\sqb{\bbE\sqb{u(Y(a))\mid X}\mid H}.
\end{align*}
The last expression does not depend on $A$, so $Y(a)\perp A\mid H$. Consistency and positivity then give the conditional-mean identity.
\end{proof}

A balancing representation identifies an average conditional on $H$. Recovering the more detailed target $\bbE\sqb{Y(a)\mid X=x}$ from matching on $H$ requires the conditional expected outcome to depend on $X$ through $H$.

\begin{proposition}[Outcome information retained by the embedding]
\label{prop:outcome-sufficient-embedding}
Suppose that there is a function $m_a$ such that
\[
\bbE\sqb{Y(a)\mid X}
=
m_a(\varphi(X))
\]
almost surely. Under the conditions of Proposition~\ref{prop:identification-x}, it holds that
\[
f_0(a,x)=m_a(\varphi(x))
\]
for $P_X$-almost every $x$.
\end{proposition}

\begin{proof}
Proposition~\ref{prop:identification-x} gives $f_0(a,X)=\bbE\sqb{Y(a)\mid X}$. Substituting the stated condition proves the result.
\end{proof}

The two representation conditions have different roles. Balancing concerns adjustment for action assignment. The last proposition concerns the heterogeneity needed to retain the same conditional expected outcomes as the original covariates. The construction of \citet{Imai2026causalinference} combines structured covariates with features extracted from unstructured data. The estimates studied by \citet{Dasanaike2026usingembedding} can supply a probabilistic measurement of an unavailable attribute, but they do not imply either condition without further assumptions.

The variables used for adjustment need not all appear in the final policy. Suppose $X=(Q,W)$ and the policy is restricted to depend on $Q$. After the conditional expected outcomes are adjusted using $(Q,W)$, the value relevant to a $Q$-only policy is based on
\[
f_0^Q(a,q)
\coloneqq
\bbE\sqb{f_0(a,Q,W)\mid Q=q}.
\]
The optimal policy within this restricted information set selects an action that maximizes $f_0^Q(a,q)$.

\section{Matching Weights and Outcome Information}
\label{app:matching-weights}
Contexts and historical actions alone do not determine the optimal action. To see this, fix any joint distribution of $(X,A)$. One possible conditional-outcome model sets $f_0(1,x)=1$ and $f_0(0,x)=0$ for every $x$, while another sets $f_0(1,x)=0$ and $f_0(0,x)=1$. Both models have the same distribution of $(X,A)$, but their optimal actions are opposite. Outcome information or an external model of expected outcomes is therefore needed.

For a neighborhood $\mathcal N_k(x)$, the local action frequency
\[
\widehat e_k(a\mid x)
=
\frac{1}{k}\sum_{j\in\mathcal N_k(x)}\mathbbm{1}\sqb{A_j=a}
\]
targets the propensity score $e_0(a\mid x)=\Pr(A=a\mid X=x)$ under standard nearest-neighbor conditions. If this estimator is consistent and the propensity has a unique maximizer, selecting the largest local frequency converges to the most likely historical action. Its limiting regret is
\[
\bbE\sqb{f_0(a^*(X),X)-f_0(a_{\mathrm{hist}}(X),X)},
\qquad
a_{\mathrm{hist}}(x)\in\argmax_{a\in\calA}e_0(a\mid x),
\]
which need not be zero.

The matching-weight identity in Section~\ref{sec:nn-regret} follows directly from changing the order of summation. If
\[
K_{a,j}
=
\sum_{i=1}^m\mathbbm{1}\sqb{
j\in\mathcal N_{k,a}(\varphi(X_i^{\mathrm{tar}}))
},
\]
then
\begin{align*}
\frac{1}{m}\sum_{i=1}^m\widetilde f_k(a,X_i^{\mathrm{tar}})
&=
\frac{1}{mk}\sum_{i=1}^m\sum_{j=1}^{N_a}
\mathbbm{1}\sqb{j\in\mathcal N_{k,a}(X_i^{\mathrm{tar}})}Y_{a,j}\\
&=
\sum_{j=1}^{N_a}\frac{K_{a,j}}{mk}Y_{a,j}.
\end{align*}
Bayes' rule gives $p_X(x)/p_a(x)=\Pr(A=a)/\Pr(A=a\mid X=x)$. This is the population relation behind the connection among match counts, density-ratio estimation, and inverse propensity weighting \citep{Lin2023estimationbased,Kato2025nearestneighbor}.

\section{Policy-Value Estimation}
\label{app:ewm}

\subsection{Value induced by an expected-outcome estimate}
Let $\widehat s_f(x)\in\argmax_{a\in g(x)}\widehat f(a,x)$ be measurable. For any $s\in\mathcal S_g$,
\[
\widehat f(\widehat s_f(x),x)
\geq
\widehat f(s(x),x)
\]
for every $x$. Taking expectations gives $\widehat V_f(\widehat s_f)\geq\widehat V_f(s)$. Applying the same inequality at each target context proves the empirical version. This is the equivalence used in
Section~\ref{sec:regret-analysis}.

\subsection{Binary EWM and least-squares CATE estimation}
Let $\mu_a(x)=\bbE\sqb{Y(a)\mid X=x}$ and $\tau(x)=\mu_1(x)-\mu_0(x)$. For a deterministic binary policy $\pi$, define $g_\pi(x)=2\pi(x)-1$. Then,
\begin{align*}
V(\pi)
&=\bbE\sqb{\mu_0(X)+\pi(X)\tau(X)}\\
&=\bbE\sqb{\mu_0(X)}
+\frac{1}{2}\bbE\sqb{\tau(X)}
+\frac{1}{2}\bbE\sqb{g_\pi(X)\tau(X)}.
\end{align*}

\begin{proposition}[EWM and least-squares CATE estimation]
\label{prop:ewm-ls-equivalence}
Let $\Pi$ be a class of deterministic binary policies and let $\mathcal G_\Pi=\cb{2\pi-1:\pi\in\Pi}$. Maximizing $V(\pi)$ over $\Pi$ is equivalent, under $g=2\pi-1$, to minimizing
\[
\bbE\sqb{(\tau(X)-g(X))^2}
\]
over $\mathcal G_\Pi$. The same algebra applies to empirical criteria when $\tau(X_i)$ is replaced by a common pseudo-outcome.
\end{proposition}

\begin{proof}
Every $g\in\mathcal G_\Pi$ satisfies $g(X)^2=1$. Therefore,
\[
\bbE\sqb{(\tau(X)-g(X))^2}
=
\bbE\sqb{\tau(X)^2}+1-2\bbE\sqb{\tau(X)g(X)}.
\]
The first two terms do not depend on $g$. Minimizing the squared loss is equivalent to maximizing $\bbE\sqb{\tau(X)g(X)}$, which is equivalent to maximizing $V(\pi)$. The empirical statement follows in the same way. This is the equivalence studied by \citet{Kato2025bridginggap}.
\end{proof}

\subsection{IPW and AIPW policy values}
Suppose that binary actions satisfy conditional exchangeability and positivity, and write $e(x)=\Pr(A=1\mid X=x)$. For a deterministic policy $\pi$, the IPW value estimator is
\[
\widehat V_{\mathrm{IPW}}(\pi)
=
\frac{1}{n}\sum_{i=1}^n
\left(
\frac{A_iY_i\pi(X_i)}{\widehat e(X_i)}
+
\frac{(1-A_i)Y_i(1-\pi(X_i))}{1-\widehat e(X_i)}
\right).
\]
With the true propensity score, iterated expectations give $\bbE\sqb{\widehat V_{\mathrm{IPW}}(\pi)}=V(\pi)$. Let $\widehat\mu_a(x)$ be expected-outcome estimates. The AIPW value estimator is
\begin{align*}
\widehat V_{\mathrm{AIPW}}(\pi)
=\frac{1}{n}\sum_{i=1}^n
\Bigg(&\widehat\mu_{\pi(X_i)}(X_i)\\
&+
\frac{\mathbbm{1}\sqb{A_i=\pi(X_i)}}{\widehat e(A_i\mid X_i)}
\p{Y_i-\widehat\mu_{A_i}(X_i)}
\Bigg),
\end{align*}
where $\widehat e(1\mid x)=\widehat e(x)$ and $\widehat e(0\mid x)=1-\widehat e(x)$. The second term uses the observed outcome to correct the expected-outcome estimate. This is the statistical difference discussed in
Section~\ref{sec:discussion}.

\section{Nonparametric Rates for Transformer Expected-Outcome Estimates}
\label{app:transformer-rates}
The regret theorems take a prediction-error bound as an input. This section records the substitution for the transformer results cited in the main text and keeps the representation dimension, in-context sample size, and number of pretraining tasks separate.

The analysis of \citet{Kim2024transformersare} considers a transformer composed of a deep neural network with $N_{\mathrm{rep}}$-dimensional output and one linear-attention layer. In the Besov setting, one of its main risk bounds has the form
\[
q_{\mathrm{ICL}}
\lesssim
N_{\mathrm{rep}}^{-2\alpha/d}
+
\frac{N_{\mathrm{rep}}\log N_{\mathrm{rep}}}{n_{\mathrm{ctx}}}
+
\frac{N_{\mathrm{rep}}^2\log N_{\mathrm{rep}}}{T_{\mathrm{pre}}}.
\]
The three terms are the approximation error, the in-context generalization error, and the pretraining generalization error. When $T_{\mathrm{pre}}$ is sufficiently large and $N_{\mathrm{rep}}\asymp n_{\mathrm{ctx}}^{d/(2\alpha+d)}$, the first two terms give the minimax MSE rate $n_{\mathrm{ctx}}^{-2\alpha/(2\alpha+d)}$ up to logarithmic factors. \citet{Oko2024pretrainedtransformer} studies adaptation to low-dimensional target functions, while \citet{Ching2026efficientand} gives a transformer construction attaining the H\"older minimax MSE rate and states a separate requirement on the number of pretraining sequences. Related approximation and estimation results are given by \citet{Takakura2023approximationand} and \citet{Havrilla2024understandingscaling}.

Suppose that the two binary expected-outcome estimates satisfy
\[
\bbE\sqb{(\widehat f(a,X)-f_0(a,X))^2}
\leq
q_{a,n},
\qquad
a\in\cb{0,1}.
\]
Then,
\[
\bbE\sqb{(\widehat\tau(X)-\tau_0(X))^2}
\leq
2q_{0,n}+2q_{1,n}.
\]
The MSE statement of Theorem~\ref{thm:binary-regret} therefore gives
\[
R(\widehat a)
\leq
C(q_{0,n}+q_{1,n})^{(1+\kappa)/(2+\kappa)}.
\]
Substituting the displayed in-context risk yields the corresponding regret bound. If the prediction result is instead a pointwise exponential-deviation inequality with root scale $r_n$, the second statement of Theorem~\ref{thm:binary-regret} gives a bound of order $r_n^{1+\kappa}$. The MSE result alone does not imply that sharper bound.

The sampling model must also match the RAG application. The in-context examples in these analyses are sampled under the task model used for pretraining, while RAG retrieves examples because they are close to the query. A direct application therefore requires the theoretical sampling condition to cover this selected context. Section~\ref{sec:nn-regret} instead analyzes nearest-neighbor retrieval directly. For a general trained generator, the prediction error relative to the required conditional expected outcome must be evaluated or assumed separately.

\section{Details of the Simulation Studies}
\label{app:simulation-designs}

\subsection{Prepared Natural-Language Queries and Case Reports}
All natural-language documents are completed before notebook execution. Each DGP uses 100 manually written source templates. Within every repetition, DGP~1 and DGP~2 use each source pattern five times in the historical corpus, while DGP~3 uses each source pattern six times. Ten distinct source patterns are used for the ten test queries in a repetition.

Each historical record contains the covariates together with the observed action and outcome. Each test record contains covariates and a decision question. The outcome is written as a final performance score to one decimal place. The retrieval profile is stored separately from the prose document. It lists the attributes in a fixed order but excludes the action and outcome as well as the decision question. The numerical state is not included in the text supplied to the language model.

For DGP~1 and DGP~2, the levels in $\cb{-2,-1,0,1,2}$ are expressed as categories from very low to very high. Experience is also expressed in years. For DGP~3, the levels in $\cb{0.1,0.3,0.5,0.7,0.9}$ are expressed as categories from very low to very high. The same completed corpus and test queries are used by every method within a repetition.

\subsection{DGP 1: Linear Conditional Effects with Observational Assignment}
Let $X=(X_1,\ldots,X_6)$, where the components are independently and uniformly distributed on $\cb{-2,-1,0,1,2}$. They represent recent performance; workload; task complexity; experience; schedule flexibility; and team support. The two actions are
\[
0:\text{ temporary workload reduction}
\qquad
1:\text{ weekly individual coaching}.
\]
We define
\[
\begin{aligned}
\mu_0(x)
&=65+4x_1-3x_2-2x_3+2x_4+1.5x_6,\\
\tau_0(x)
&=5x_1-4x_2+3x_5,\\
\mu_1(x)
&=\mu_0(x)+\tau_0(x).
\end{aligned}
\]
The historical probability of action~1 is
\[
e_0(x)
=\operatorname{clip}_{[0.10,0.90]}
\left(
\operatorname{logit}^{-1}
\left(0.5x_2-0.4x_4+0.3x_6-0.2x_1\right)
\right).
\]
We draw $A\mid X=x\sim\operatorname{Bernoulli}(e_0(x))$ and generate
\[
Y=\mu_A(X)+\epsilon,
\qquad
\epsilon\sim\mathcal N(0,3^2).
\]
All variables entering the propensity and the two conditional expected outcomes are present in the text.

\subsection{DGP 2: Smooth Nonlinear Personalization with Randomized Assignment}
DGP~2 uses the same state support and the same two actions. We define
\[
\begin{aligned}
\mu_0(x)
={}&65+3\sin\left(\frac{\pi x_1}{4}\right)-2x_2-1.2x_3+1.4x_4+1.2x_6\\
&+0.5x_2x_6-0.4x_4^2
\end{aligned}
\]
and
\[
\begin{aligned}
b(x)
={}&1.25x_1-1.50x_2+x_3-x_4+0.75x_5-0.75x_6\\
&+0.80\sin\left(\frac{\pi(x_1-x_2)}{4}\right)
+0.60\sin\left(\frac{\pi(x_3-x_4)}{4}\right),\\
\tau_0(x)
={}&8\tanh\left(\frac{b(x)}{4}\right),\\
\mu_1(x)
={}&\mu_0(x)+\tau_0(x).
\end{aligned}
\]
We draw $A\sim\operatorname{Bernoulli}(1/2)$ independently of $X$ and generate
\[
Y=\mu_A(X)+\epsilon,
\qquad
\epsilon\sim\mathcal N(0,2.5^2).
\]
The distribution of $X$ is symmetric and $b(-x)=-b(x)$. It follows that $\bbE[\tau_0(X)]=0$. Both actions are optimal on substantial parts of the state space, while any policy that receives no current state is restricted to a population-level choice. The conditional effect is smooth in all six variables used to construct the retrieval profile.

\subsection{DGP 3: Large Finite Action Space}
Let
\[
\mathcal G=\cb{0.1,0.3,0.5,0.7,0.9}.
\]
The context is $X=(N,I,D,B)$, where $N=(N_1,N_2,N_3,N_4)\in\mathcal G^4$ describes the need for technical skill development; workload relief; managerial guidance; and peer support. The variables $I$, $D$, and $B$ in $\mathcal G$ describe suitable intervention intensity; suitable duration; and baseline performance. All seven variables are drawn independently and uniformly from $\mathcal G$.

An action is a triple $a=(t,\ell,r)$. The intervention type $t$ is individual coaching; temporary workload reduction; technical training; or peer support. The intensity level is $\ell\in\cb{1,2,3}$ and the duration is $r\in\cb{2,4}$ weeks. The four type profiles are
\[
\begin{aligned}
v_{\mathrm{coaching}}&=(0.25,0.10,1.00,0.20),\\
v_{\mathrm{workload}}&=(0.05,1.00,0.20,0.05),\\
v_{\mathrm{training}}&=(1.00,0.05,0.35,0.15),\\
v_{\mathrm{peer}}&=(0.20,0.15,0.45,1.00).
\end{aligned}
\]
The complete action space contains 24 actions. We set
\[
c(a)=0.35\ell+0.15r
\]
and define
\[
f_0(a,x)
=62+3B+\frac{9}{2.2}N^\top v_t
-4\left(I-\frac{\ell}{3}\right)^2
-3\left(D-\frac{r}{4}\right)^2
-c(a).
\]
Historical actions are generated from the misspecified score
\[
s(a,x)=2N^\top v_t-0.7c(a).
\]
The assignment probability is
\[
\Pr(A=a\mid X=x)
=0.20
\frac{\exp(s(a,x))}{\sum_{a'\in\calA}\exp(s(a',x))}
+\frac{0.80}{24}.
\]
The assignment rule omits suitable intensity and duration as well as baseline performance, even though all three enter $f_0$. The realized outcome is
\[
Y=f_0(A,X)+\epsilon,
\qquad
\epsilon\sim\mathcal N(0,2.5^2).
\]
Each action has a stable identifier from \texttt{A01} through \texttt{A24} and a natural-language description.

\subsection{Embedding and Retrieval}
The main experiments use the OpenAI embedding model \texttt{text-embedding-3-small}. Every vector is normalized and cosine distance is used. The embedding input is the fixed-order retrieval profile rather than the full prose report. This profile contains all attributes but excludes the historical action and observed outcome as well as the decision question. The generator receives the full prose query and the full retrieved case reports.

For DGP~1 and DGP~2, the experiment retrieves the six nearest reports observed under action~0 and the six nearest reports observed under action~1. One-step RAG-PL and Two-step RAG-PL receive exactly the same query and reports. Reports are displayed by retrieval rank. The action shown first alternates across ranks, and the initial action is randomized by query. The action descriptions and the two score summaries follow the same randomized order. The prompt states that this order has no priority.

For DGP~3, the initial retrieval returns the twelve nearest reports without restricting the historical action. Both procedures receive these reports and the complete action catalog. Two-step RAG-PL then retrieves five reports within each generated candidate action.

\subsection{Implementation of the RAG Procedures}
The main experiments use the model identifier \texttt{gpt-5.4-mini}. The temperature is set to $0$, and the reasoning effort is set to \texttt{none}. For DGPs~1 and~2, the request configuration allows at most 16,384 input tokens and 256 output tokens. For DGP~3, the corresponding limits are 32,768 input tokens and 1,200 output tokens. Each test query is processed independently, and no conversation history is carried across test queries. All model outputs are validated against predefined structured-output schemas.

In DGPs~1 and~2, every method receives six displayed reports under each action and the exact arithmetic mean of the six displayed outcomes under each action. These summaries are deterministic functions of the displayed reports and therefore add no observations. Providing the summaries to every method ensures that differences among the methods are not driven merely by the model's ability to calculate sample averages from the displayed outcomes.

One-step RAG-PL returns only a selected action. Two-step RAG-PL returns one estimated conditional expected outcome under each action. The experiment code selects action~1 when its estimated conditional expected outcome is greater than or equal to that under action~0. Thus, action~1 is selected under an exact tie.

RAG without covariates receives neither a current query profile nor historical pre-action profiles. Its historical reports contain only the action and observed outcome. Within each repetition, six reports are sampled uniformly without replacement under each action. The model receives the exact arithmetic mean of the six displayed outcomes under each action and returns two population-level expected-outcome estimates. The experiment code applies the same tie-breaking rule as above. Because the input to this baseline does not vary across the test queries within a repetition, the baseline is queried once for each retrieval database, and its selected action is applied to every test query in that repetition.

In DGP~3, both RAG-PL procedures receive the same current query, the same twelve initially retrieved reports, and the complete catalog of 24 actions. One-step RAG-PL returns one valid action identifier from the catalog. Two-step RAG-PL first returns five distinct valid action identifiers. The procedure then retrieves five reports observed under each candidate action. A subsequent candidate-evaluation request supplies the five candidate actions and their matched reports and asks for one conditional expected-outcome estimate for each candidate. The experiment code selects the candidate with the largest estimate, using a fixed rule to break ties.

During candidate evaluation, valid estimates contained in an incomplete response are retained, and only the missing candidate estimates are requested in subsequent attempts. Candidate generation and candidate evaluation each permit at most eight attempts in total. All experiments reported below obtained complete valid outputs within these limits.

The DGP~3 comparison is not matched in either the number of model requests or the amount of retrieved evidence. In particular, Two-step RAG-PL uses an additional candidate-generation stage, candidate-specific retrieval, and a candidate-evaluation stage. We therefore interpret the DGP~3 results as a comparison of the implemented end-to-end pipelines, rather than as an isolated comparison of one-step and two-step action selection under equal computational budgets.

\subsection{Evaluation}
Each DGP contains 20 independent repetitions and ten test queries per repetition. DGP~1 and DGP~2 contain 500 historical reports per repetition. DGP~3 contains 600. All methods in a repetition use the same completed corpus and queries. All 20 repetitions were completed for every method and were included in the main summaries.

For each test query, we record the selected and optimal actions. We also record their true values and the resulting regret. An action is counted as optimal whenever its true conditional expected outcome equals the maximum, so ties are handled by value rather than by agreement with one tie-breaking label. Query-level quantities are averaged within a repetition. Means and standard errors are then calculated across repetitions.

For methods that return expected outcomes, we also calculate expected-outcome mean-squared error. In the binary DGPs, we calculate the squared error of the estimated contrast. For DGP~3, we calculate candidate-specific expected-outcome mean-squared error and centered mean-squared error. We also calculate pairwise ranking accuracy.

For DGP~3, let $g(x)$ be the five generated candidates and let $a_g^*(x)$ be the best action in this set. We record
\[
f_0(a^*(x),x)-f_0(a_g^*(x),x)
\]
and
\[
f_0(a_g^*(x),x)-f_0(\widehat a_{\mathrm{two}}(x),x).
\]
Their sum equals the query-level total regret. We also record exact candidate coverage. Approximate coverage is recorded at value-gap thresholds of $0.1$, $0.25$, and $0.5$.

\subsection{Language-Model Robustness}
After all three GPT-5.4~mini experiments are completed, DGP~1 is repeated with Qwen~3.5~4B and Gemma~3~4B. The local models use the same completed documents and OpenAI retrieval vectors as the main DGP~1 experiment. They also use the same matched reports and method definitions. They are used only as RAG generators. Results are reported separately by model because the purpose is to check the direction of the method comparison rather than rank language models.

\subsection{Additional Results}
\label{app:simulation-results}
Table~\ref{tab:binary-estimation} reports the expected-outcome errors for procedures that return numerical estimates in the binary DGPs. Two-step RAG-PL has lower expected-outcome and contrast mean-squared error than RAG without covariates in both DGPs.

\begin{table}[t]
\centering
\caption{Expected-outcome estimation in the binary DGPs}
\label{tab:binary-estimation}
\small
\begin{tabular}{lrr}
\toprule
Method & Expected-outcome MSE & Contrast MSE \\
\midrule
\multicolumn{3}{l}{\textbf{DGP 1}} \\
Two-step RAG-PL & 111.4406 & 161.4745 \\
RAG without covariates & 229.4167 & 251.6467 \\
\midrule
\multicolumn{3}{l}{\textbf{DGP 2}} \\
Two-step RAG-PL & 28.5777 & 31.1131 \\
RAG without covariates & 51.9164 & 41.3173 \\
\bottomrule
\end{tabular}
\end{table}

Table~\ref{tab:dgp3-components} separates the regret of Two-step RAG-PL in DGP~3. Candidate generation accounts for mean regret $0.4822$, while selection within the generated set accounts for $0.7332$. Pairwise ranking accuracy within the candidate set is $0.6200$, and the best candidate is selected with probability $0.3450$.

\begin{table}[t]
\centering
\caption{Regret decomposition and candidate evaluation in DGP 3}
\label{tab:dgp3-components}
\small
\begin{tabular}{lrr}
\toprule
Quantity & Mean & Standard error \\
\midrule
Total regret & 1.2154 & 0.0837 \\
Candidate-set regret & 0.4822 & 0.0387 \\
Within-candidate regret & 0.7332 & 0.0616 \\
Exact candidate coverage & 0.4450 & 0.0294 \\
Candidate coverage within $0.1$ & 0.4500 & 0.0286 \\
Candidate coverage within $0.25$ & 0.5200 & 0.0258 \\
Candidate coverage within $0.5$ & 0.6700 & 0.0309 \\
Candidate expected-outcome MSE & 3.0358 & 0.1412 \\
Centered candidate outcome MSE & 1.9420 & 0.1096 \\
Pairwise ranking accuracy & 0.6200 & 0.0153 \\
Best-candidate selection probability & 0.3450 & 0.0359 \\
\bottomrule
\end{tabular}
\end{table}

Table~\ref{tab:local-model-robustness} gives the DGP~1 results for the two local generators. For both models, One-step RAG-PL and Two-step RAG-PL have lower mean regret than RAG without covariates. The relative ordering of the one-step and two-step procedures varies with the generator, so these results are used as a robustness check rather than a model ranking.

\begin{table}[t]
\centering
\caption{DGP 1 robustness results with local language models}
\label{tab:local-model-robustness}
\small
\begin{tabular}{lrrrr}
\toprule
Method & Regret & Regret SE & Policy value & Optimal-action probability \\
\midrule
\multicolumn{5}{l}{\textbf{Qwen 3.5 4B}} \\
One-step RAG-PL & 3.4600 & 0.4353 & 66.6025 & 0.5450 \\
Two-step RAG-PL & 3.6050 & 0.4418 & 66.4575 & 0.5250 \\
RAG without covariates & 4.5450 & 0.3943 & 65.5175 & 0.4100 \\
\midrule
\multicolumn{5}{l}{\textbf{Gemma 3 4B}} \\
One-step RAG-PL & 3.1450 & 0.4055 & 66.9175 & 0.5700 \\
Two-step RAG-PL & 3.5200 & 0.4070 & 66.5425 & 0.5450 \\
RAG without covariates & 4.5450 & 0.3943 & 65.5175 & 0.4100 \\
\bottomrule
\end{tabular}
\end{table}

\end{document}